\documentclass[12pt]{article}
\usepackage{lMac}
\usepackage{becMac}
\newcommand{\brphi}{{\breve\phi}}

\newcommand{\lb}{\mathrm{lb}}

\def\showintremarks{n}   
\def\showissues{n}       

\begin{document}
\title{The Small Field Parabolic Flow for Bosonic Many--body Models:\\
        \Large Part 4 --- Background and Critical Field Estimates}

\author{Tadeusz Balaban}
\affil{\small Department of Mathematics \authorcr
       Rutgers, The State University of New Jersey \authorcr
       tbalaban@math.rutgers.edu\authorcr
       \  }

\author{Joel Feldman\thanks{Research supported in part by the Natural 
                Sciences and Engineering Research Council 
                of Canada and the Forschungsinstitut f\"ur 
                Mathematik, ETH Z\"urich.}}
\affil{Department of Mathematics \authorcr
       University of British Columbia \authorcr
       feldman@math.ubc.ca \authorcr
       http:/\hskip-3pt/www.math.ubc.ca/\squig feldman/\authorcr
       \  }

\author{Horst Kn\"orrer}
\author{Eugene Trubowitz}
\affil{Mathematik \authorcr
       ETH-Z\"urich \authorcr
       knoerrer@math.ethz.ch, trub@math.ethz.ch \authorcr
       http:/\hskip-3pt/www.math.ethz.ch/\squig knoerrer/}


\maketitle

\begin{abstract}
\noindent
This paper is a contribution to a program to see symmetry breaking in a
weakly interacting many Boson system on a three dimensional lattice at 
low temperature.  It is part of an analysis of the ``small field''  approximation 
to the ``parabolic flow'' which exhibits the formation of a ``Mexican hat''
potential well. Here we prove the existence of and bounds on the background
and critical fields that arise from the steepest descent attack 
that is at the core of the renormalization group step anaylsis of \cite{PAR1,PAR2}.

\end{abstract}

\newpage
\tableofcontents

\newpage
\section{Introduction}
In \cite{PAR1,PAR2}, we use the block spin renormalization group
formalism to exhibit the formation\footnote{in the
small field regime} of a potential well, signalling the 
onset of symmetry breaking in a many particle system of weakly 
interacting Bosons in three space dimensions. For an overview, 
see \cite{ParOv}. For a brief discussion of the algebraic aspects 
of the block spin method see \cite{BlockSpin}.

In \cite{ParOv,PAR1,PAR2} the model is initially formulated as a 
functional integral with integration variables indexed by the 
lattice\footnote{Of course $\cX_0$ is a finite set and so is perhaps
more accurately described as a discrete torus, rather than a
lattice.}$^{,}$\footnote{In this introduction, we are only going to
give ``impressionistic'' definitions. The detailed, technically
complete, definitions are given in \cite[Appendix \appDefinitions]{PAR1}.
Specifically, for the lattices, see \cite[\S\appDEFlattices]{PAR1}.} 
\begin{equation*}
\cX_0=\big(\bbbz/L_\tp\bbbz\big)\times
         \big(\bbbz^3/L_\sp\bbbz^3\big) 
\end{equation*} 
$\cX_0$ is a unit lattice in the sense that the distance 
between nearest neighbours in the lattice is $1$. 
During each renormalization group step this lattice
is scaled down. In each of the first $\np$ steps, which are the
steps considered in \cite{ParOv,PAR1,PAR2}, we use 
(anisotropic) ``parabolic scaling'' which decreases
the lattice spacing in the temporal direction 
by a factor of $L^2$ and in the spatial directions by a factor of $L$. 
Here $L\ge 3$ is a fixed odd natural number.
So after $n$ renormalization group steps the lattice spacing 
in the spatial directions is $\veps_n=\frac{1}{L^n}$ and in the 
temporal direction is  $\veps_n^2=\frac{1}{L^{2n}}$
and the lattice $\cX_0$  has been scaled down to
\begin{equation*}
\cX_n= \big(\sfrac{1}{L^{2n}}\bbbz \big/ \sfrac{L_\tp}{L^{2n}}\bbbz\big)
   \times \big(\sfrac{1}{L^n}\bbbz^3 \big/\sfrac{L_\sp}{L^n}\bbbz^3 \big)
\end{equation*}
We call $\cX_n$ the ``$\veps_n$--lattice''.

The dominant ``pure small field'' part of the original functional 
integral representation of this model is, after $n$ renormalization 
group steps, reexpressed as a functional integral 
$
\int \smprod_{x\in\cX_0^{(n)}} \sfrac{d\psi^*(x) d\psi(x)}{2\pi i}\,
 e^{{\rm Action}_n}
$
with integration variables indexed by the unit sublattice 
\begin{align*}
\cX_0^{(n)}&=
     \big(\bbbz/\sfrac{L_\tp}{L^{2n}}\bbbz\big)
     \times\big(\bbbz^3/\sfrac{L_\sp}{L^n}\bbbz^3\big)
\end{align*}
of $\cX_n$. More generally, we have to deal with the 
decreasing sequence of sublattices
\begin{align*}
\cX_j^{(n-j)}&=
     \big(\sfrac{1}{L^{2j}}\bbbz/\sfrac{L_\tp}{L^{2n}}\bbbz\big)
     \times\big(\sfrac{1}{L^j}\bbbz^3/\sfrac{L_\sp}{L^n}\bbbz^3\big)
\end{align*}
of $\cX_n$.
The lower index gives the ``scale'' of the lattice. 
That is, the distance between nearest neighbour points of the lattice. 
The upper index  determines the number of points 
in the sublattice $\big($namely 
$\big(\sfrac{L_\tp}{L^{2(n-j)}}\big)\big(\sfrac{L_\sp}{L^{n-j}}\big)^3$
$\big)$. The sum of the upper and lower indices gives the number of 
the renormalization group step. For fields $\phi,\psi$ on $\cX_j^{(n-j)}$,
we use the ``real'' inner product
$
\<\phi,\psi\>_j =\sfrac{1}{L^{5j}}\sum_{u\in\cX_j^{(n-j)}}\phi(u)\psi(u)
$.
The vector space $\bbbc^{\cX_j^{(n-j)}}$, equipped with the inner product
$\<\phi^*,\psi\>$, is a Hilbert space, which we denote $\cH_j^{(n-j)}$.

\noindent
Roughly speaking, in each block spin RG step one
\begin{itemize}[leftmargin=*, topsep=2pt, itemsep=0pt, parsep=0pt]
\item[$\circ$] 
paves $\cX_0^{(n)}$ by rectangles centered at the points of the
sublattice $\cX_{-1}^{(n+1)}\subset \cX_0^{(n)}$ and then, 
\item[$\circ$] 
for each $y\in\cX_{-1}^{(n+1)}$, integrates out all values of $\psi$ 
whose ``average value'' over the rectangle centered at $y$ is equal to the 
value of a given field $\th(y)$ on $\cX_{-1}^{(n+1)}$.
The precise ``average value'' used is determined by an averaging 
profile\footnote{In \cite{PAR1,PAR2}, the averaging profile is an 
iterated convolution of the characteristic function of the rectangle 
with itself. See \cite[\S\appDEFblockspinops]{PAR1}}.
One uses this profile to define\footnote{For the detailed definition 
of the averaging operator $Q$, see \cite[\S\appDEFblockspinops]{PAR1}.} 
an averaging operator $Q$ from  the 
space of fields on $\cX_0^{(n)}$  to the space  of fields on 
$\cX_{-1}^{(n+1)}$. One then implements the ``integrating out'' by first 
inserting into the integrand $1$, expressed as a constant times the Gaussian integral
\begin{equation}\label{eqnBGEgaussianOne}
\int \smprod_{y\in\cX_{-1}^{(n+1)}}\sfrac{d\th^*(y) d\th(y)}{2\pi i}
                            e^{-a\< \th^*-Q\,\psi_*\,,\, \th-Q\,\psi \>_{-1}} 
\end{equation}
with some constant $a>0$, and then interchanging the order of the $\th$ and
$\psi$ integrals.
\end{itemize}
We use stationary phase/steepest descent to control these integrals.
This naturally leads one to express the action not solely
in terms of the integration variables $\psi$, but also
in terms of ``background fields'', which are concatinations of 
``steepest descent'' critical field maps for all previous steps.
See \cite[Remark \remBSactionRecur\ and 
Proposition \propBSconcatbackgr.c]{BlockSpin}.
The dominant part of the action is then of the form
\begin{equation*}
A_n(\psi_*,\psi,\phi_*,\phi,\mu,\cV)
   \Big|_{\atop{\phi_* = \phi_{*n}(\psi^*,\psi)}
               {\phi = \phi_n(\psi^*,\psi)}}                       
\end{equation*}
where
\begin{equation}\label{eqnBGErepresentationaction}
\begin{split}
&A_n(\psi_*,\psi,\phi_*,\phi,\mu,\cV)
 =-\<\psi_*-Q_n\phi_*,\fQ_n\big(\psi-Q_n\phi\big)\>_0
 -\< \phi_*,\,D_n\phi\>_n \\
&\hskip3in     - \cV(\phi_*,\phi)
+\mu \< \phi_*,\,\phi\>_n              
\end{split}
\end{equation}
and

\begin{itemize}[leftmargin=*, topsep=2pt, itemsep=0pt, parsep=0pt]
\item[$\circ$] 
$Q_n:\bbbc^{\cX_n}\rightarrow\bbbc^{\cX_0^{(n)}}$ 
is an averaging operator that is the composition of the 
averaging operations for all previous steps. For the precise
definition of $Q_n$, see \cite[\S\appDEFblockspinops]{PAR1}. For bounds
on $Q_n$, see \cite[Remark \remPBSqnft.a and Lemma \lemPBSunppties]{POA}.

\item[$\circ$]
the term $\< \psi^*-Q_n\,\phi_*\,,\, \fQ_n(\psi-Q_n\,\phi) \>_0 $
is a residue of the exponents in the  Gaussian integrals 
\eqref{eqnBGEgaussianOne} inserted in the
previous steps. The operator $\fQ_n$ is bounded and boundedly invertible.
For the precise definition of $\fQ_n$, see \cite[\S\appDEFblockspinops]{PAR1}. 
See \cite[Remark \remBSactionRecur]{BlockSpin} for the
recursion relation that builds $\fQ_n$.
For bounds on $\fQ_n$, see \cite[Remark \remPBSqnft.c and 
Proposition \propPBSAnppties]{POA}.

\item[$\circ$] 
$D_n$ is a discrete differential operator. It is simply a scaled
version of the discrete differential operator that appeared in
the initial action, which, in turn, was built from the
single particle ``kinetic energy'' operator. Think of $D_n$
as behaving like $-\partial_0-\De$. For the detailed definition
of $D_n$, see \cite[\S\appDEFoperators]{PAR1}. Various properties of and bounds
on $D_n$ are provided in \cite[\S\secPOdiffOps]{POA}.

\item[$\circ$] 
$\cV$ is an interaction. It is a quartic monomial
\begin{equation*}
\cV(\phi_*,\phi)=\half\int_{\cX_n^4} du_1 \cdots du_4\  
                  V(u_1,u_2,u_3,u_4)\,  
                  \phi_*(u_1) \phi(u_2)\phi_*(u_3) \phi(u_4)
\end{equation*}
where $\int_{\cX_n} du = \frac{1}{L^{5n}}\sum_{u\in\cX_n}$
and the kernel $V(u_1,u_2,u_3,u_4)$ is translation invariant
and exponentially decaying.

\item[$\circ$] 
$\mu$ is a chemical potential. In this paper, we are interested
in $\mu>0$ that are sufficiently small.  For more details, see 
\cite[Theorem \thmTHmaintheorem]{PAR1}.

\item[$\circ$] 
The background fields\footnote{We routinely use the ``optional $*$''
notation $\al_{(*)}$ to denote ``$\al_*$ or $\al$''. 
} $\phi_{(*)n}(\psi_*,\psi,\mu,\cV)$, in
addition to being concatinations of ``steepest descent'' critical 
field maps for all previous steps, are critical points for the map
$$
(\phi_*,\phi) \mapsto A_n(\psi_*,\psi,\phi_*,\phi,\mu,\cV)
$$
\end{itemize}
In this paper we fix an integer $1\le n\le \np$, where $\np$ is the number
of ``parabolic scaling'' renormalization group steps considered
in \cite{PAR1,PAR2}, and prove  existence and properties 
of the background fields as above, in the concrete setting of 
\cite{PAR1,PAR2}.
By definition, they are solutions of the ``background field equations''
\begin{equation*}
\sfrac{\partial\hfill}{\partial\phi_*}A_n(\psi_*,\psi,\phi_*,\phi,\mu,\cV)
=\sfrac{\partial\hfill}{\partial\phi}A_n(\psi_*,\psi,\phi_*,\phi,\mu,\cV)
=0
\end{equation*}
or
\begin{equation}\label{eqnBGEbgeqns}
\begin{split}
S^*_n(\mu)^{-1}\phi_* +\cV'_*(\phi_*,\phi,\phi_*)
       &=Q_n^* \fQ_n\psi_*\\
S_n(\mu)^{-1}\phi+\cV'(\phi,\phi_*,\phi)
              &=Q_n^* \fQ_n\psi
\end{split}
\end{equation}
where\footnote{The number of RG steps, $n_p$, is chosen so that, for the
chemical potentials $\mu$ under consideration, the operator 
$D_n+Q_n^*\fQ_n Q_n-\mu$ is invertible.
}
\begin{equation*}
S_n(\mu)=(D_n+Q_n^*\fQ_n Q_n-\mu)^{-1}
\end{equation*}
and
\begin{align*}
\cV'_*(u;\ze_{*1},\ze,\ze_{*2})
&=\int du_1 du_2 du_3\  V(u_1,u_2,u_3,u)\, \ze_{*1}(u_1)\ze(u_2)\ze_{*2}(u_3)\\
\cV'(u;\ze_1,\ze_*,\ze_2)
&=\int du_2 du_3 du_4\  V(u,u_2,u_3,u_4)\, \ze_1(u_2)\ze_*(u_3)\ze_2(u_4)
\end{align*}
We also write $S_n=S_n(0)=(D_n+Q_n^*\fQ_n Q_n)^{-1} $.

In \S\ref{secBGEbgdfld} we write these equations as a fixed point 
equation and use the variant of the Banach fixed point theorem 
developed in \cite{SUB}, and summarized in Proposition \ref{SUB:propBGEeqnsoln},
 to control them. We also  show, 
in Proposition \ref{propBGEphivepssoln}, that
\begin{align*}
\phi_{(*)n}(\psi_*,\psi,\mu,\cV)
     &=S_n(\mu)^{(*)}Q_n^* \fQ_n\,\psi_{(*)}
      +\phi_{(*)n}^{(\ge 3)}(\psi_*,\psi,\mu,\cV)
\end{align*}
where 
$\,
\phi_{(*)n}^{(\ge 3)}
\, $
are analytic maps in $(\psi_*,\psi)$ from a neighbourhood of the origin in
$\,\bbbc^{\cX_0^{(n)}}\times \bbbc^{\cX_0^{(n)}}\,$ 
to $\,\bbbc^{\cX_n}\,$, and, in Corollary \ref{corBGEexpandaldephi}, that
\begin{equation}\label{eqnBGEbglocal}
\phi_{(*)n}(\psi_*,\psi,\mu,\cV)(u)
=\sfrac{a_n}{a_n-\mu}\psi_{(*)}\big(X(u)\big)
           +\brphi_{(*)n}\big((\psi_*,\{\partial_\nu\psi_*\})\,,\,
                       (\psi,\{\partial_\nu\psi\})\,,\,\mu,\cV\big)(u)
\end{equation}
where, for each point $u$ of the fine lattice $\cX_n$,  
$X(u)$ denotes the point of the unit lattice $\cX_0^{(n)}$ nearest to $u$, 
$a_n=a\big(1 +\smsum_{j=1}^{n-1}\sfrac{1}{L^{2j}}\big)^{-1}$
and
$\,
\brphi_{(*)n}
\, $
are analytic maps.
\begin{remark}\label{remBGEcnstFld}
When the fields $\psi_{(*)}$ and $\phi_{(*)}$ happen to be constant, then, 
by \cite[Remark \remMXconsteigen]{PAR2}, the equations 
\eqref{eqnBGEbgeqns} reduce to
\begin{equation}\label{eqnBGEbgeqnsCst}
\begin{split}
(a_n-\mu)\phi_* +\rv\phi_*^2\phi
       &=a_n\psi_*\\
(a_n-\mu)\phi+\rv\phi_*\phi^2
              &=a_n\psi
\end{split}
\end{equation}
where
$
\rv = \int_{\cX_n^3} dx_1 \cdots dx_3\  V(0,x_1,x_2,x_3)
$
is the average value of the kernel of $\cV$. As long as 
$\rv(|\psi_*|+|\psi|)^2$ is small enough, this system has a 
unique solution with 
\begin{equation*}
\phi_*=\sfrac{a_n}{a_n-\mu}\psi_*+O\big(\rv(|\psi_*|+|\psi|)^3\big)\qquad
\phi=\sfrac{a_n}{a_n-\mu}\psi+O\big(\rv(|\psi_*|+|\psi|)^3\big)
\end{equation*}
If $\psi_*=\psi^*$, then the solution $\phi_*=\phi^*$.
\end{remark}

In \S\ref{secBGEvrnspsi}, we prove, in Proposition \ref{propBGEdephisoln}, bounds
on maps which describe the variations of the 
background field  with respect to $\psi$. 

In \S\ref{secBGEvrnsmu}, we consider variations of the background field with 
respect to the chemical potential $\mu$ and interaction $\cV$. 
We prove, in  Proposition \ref{propBGEdephidemu}, bounds on 
\begin{equation*}
\De\phi_{(*)n}(\psi_*,\psi,\mu,\de\mu,\cV,\de\cV)
=\phi_{(*)n}(\psi_*,\psi,\mu+\de\mu,\cV+\de\cV) -\phi_{(*)n}(\psi_*,\psi,\mu,\cV)
\end{equation*}
as well as on $\partial_\nu$ and $D_n^{(*)}$ applied to these field maps.

Finally, in 
\S\ref{secBGEcritfld} we apply these results and  
\cite[Proposition \propBSconcatbackgr.a]{BlockSpin}
 to construct and bound the critical points, denoted $\psi_{*n}$, $\psi_n$,  of the map
\begin{equation*}
(\psi_*,\psi) \mapsto A_n(\psi_*,\psi,\phi_*,\phi,\mu,\cV)
   \Big|_{\atop{\phi_* = \phi_{*n}(\psi^*,\psi)}
               {\phi = \phi_n(\psi^*,\psi)}}   
\end{equation*}

The proofs and estimates in this paper depend heavily on bounds on operators
like $Q$, $Q_n$ and $S^{-1}_n(\mu)$, which in turn are developed in \cite{POA}.
The size of an operator is formulated in terms of a norm on its kernel.
\begin{definition}\label{defBGEoperatornorm}
Let $\cX$ and $\cY$ be sublattices of a common lattice having metric $d$,
with $\cX$ having a ``cell volume'' $\vol_\cX$ and
with $\cY$ having a ``cell volume'' $\vol_\cY$. 
For any operator $A:\bbbc^\cX\rightarrow\bbbc^\cY$, with kernel
$A(y,x)$, and for any mass $m\ge 0$, we define the norm
\begin{equation*}
\|A\|_m
=\max\Big\{\sup_{y\in\cY}\,\sum_{x\in\cX} \vol_{\cX}\ 
                  e^{m|y-x|}|A(y,x)|\ ,\ 
\sup_{x\in\cX}\,\sum_{y\in\cY} \vol_{\cY}\ 
                  e^{m|y-x|}|A(y,x)|
\Big\}
\end{equation*}
In the special case that $m=0$, this is just the usual $\ell^1$--$\ell^\infty$
norm of the kernel.
\end{definition}

Similarly, to measure the size of  a function $f:\big(\cX_j^{(n-j}\big)^r\rightarrow\bbbc$,
we introduce the weighted $\ell^1$--$\ell^\infty$ norm with mass $m\ge 0$
\begin{equation}\label{eqnBGEtreeNorm}
\|f\big(x_1,\cdots x_r\big)\|_m
=\max_{i=1\cdots,r}\ \max_{x\in \cX_j^{(n-j)}}\ 
         \sfrac{1}{L^{5j}}\hskip-15pt\sum_{\atop{x_1,\cdots,x_r\in\cX_j^{(n-j)}}
                   {x_i=x}}\hskip-15pt
          |f(x_1,\cdots,x_r)|\, e^{m\tau(x_1,\cdots,x_r)}
\end{equation}
where the tree length $\,\tau(x_1,\cdots,x_r)\,$ is the 
minimal length of a tree in $\,\cX_j^{(n-j)}\,$ that has  
$\,x_1,\cdots,x_r\,$ among its vertices. 

We use the terminology ``field map'' to designate an analytic map that assigns to
one or more fields on a finite set $\cX$ another field on a finite set $\cY$.
The most prominent examples of field maps in this paper are the background
fields $\phi_{(*)n}(\psi_*,\psi)$.  In Appendix \ref{appNormsFixedPoint},
we define norms on field maps that are constructed by summing 
norms, like \eqref{eqnBGEtreeNorm}, of the kernels in their power series expansions.
The kernel of a monomial, for example of degree $n$ in a field $\psi$,
is weighted by $\ka^n$, where $\ka$ is a ``weight factor'' assigned to
$\psi$.  For example, if $\phi(\psi)(y)=\sum\limits_{n=0}^\infty\ \sum\limits_{x_1,\cdots,x_n\in\cX} \hskip-5pt\vol_\cX^n 
         \ \phi_n\big(y;x_1,\cdots,x_n\big) \psi(x_1) \cdots\psi(x_n)$
\begin{equation*}
\tn \phi \tn =  \sum_{n}  \| \phi_n\|_m\ \ka^n
\end{equation*} 

For full definitions
of our norms, see \cite[\S\appDEFnorms]{PAR1}.

In this paper, we fix masses $\bar\fm>\fm>0$ and generic weight factors 
$\wf,\wf',\wf_\fl\ge 1$  and use the norm $\tn F\tn$ with mass $\fm$ and 
these weight factors to measure field maps $F$. 
The weight factor $\wf$ is used for the $\psi_{(*)}$'s, 
the weight factor $\wf'$ is used for the derivative fields
$\psi_{(*)\nu}$ and  
the weight factor $\wf_\fl$ is used for the fluctuation fields
$z_{(*)}$. See  Appendix \ref{appNormsFixedPoint}.

\begin{convention}\label{convBGEconstants}
The (finite number of) constants that appear in the bounds of this 
paper are consecutively labelled $\GGa_1, \GGa_2, \cdots$ or
$\rrho_1,\rrho_2,\cdots$. All of the constants $\GGa_j$, $\rrho_j$ 
are independent of $L$ and the scale index $n$. 
They depend only on the masses $\fm$ and $\bar\fm$ and the constant 
$\Gam_\op$ of \cite[Convention \convPOconstants]{POA} (with mass $m=\bar\fm$) 
and, for the $\rrho_j$'s, the $\mu_{\rm up}$ of \cite[Proposition \POGmainpos]{POA}.
We define $\GGa_\bg$ to be the maximum of the $\GGa_j$'s and $\rrho_\bg$ to be
the minimum of $\sfrac{1}{8}$ and the $\rrho_j$'s. We shall refer only
to $\GGa_\bg$ and $\rrho_\bg$, as opposed to the $\GGa_j$'s and $\rrho_j$'s, 
in \cite{PAR1,PAR2}.

\end{convention}

\newpage
\section{The Background Field}\label{secBGEbgdfld}

\noindent
The main existence result for the background field, which was
summarized in \cite[Proposition \propHTexistencebackgroundfields]{PAR1}, is

\begin{proposition}[Existence of the background field]\label{propBGEphivepssoln}
There are constants $\GGa_1$, $\rrho_1>0$  such that, if 
$
\|V\|_{\fm} \wf^2+|\mu| \le \rrho_1
$,
the following hold.
\begin{enumerate}[label=(\alph*), leftmargin=*]
\item
There exist solutions to the equations \eqref{eqnBGEbgeqns} 
for the background field. Precisely, there are field maps 
$\phi_{(*)n}^{(\ge 3)}$ such that
\begin{align*}
\phi_{(*)n}(\psi_*,\psi,\mu,\cV)
     &=S_n(\mu)^{(*)}Q_n^* \fQ_n\,\psi_{(*)}
      +\phi_{(*)n}^{(\ge 3)}(\psi_*,\psi,\mu,\cV)
\end{align*}
solves \eqref{eqnBGEbgeqns} and
\begin{equation*}
\TN \phi_{(*)n}^{(\ge 3)}\TN
\le \GGa_1 \|V\|_{\fm}\wf^3
\end{equation*}
Furthermore 
$\phi_{*n}^{(\ge 3)}$ is of degree at least one in $\psi_*$ and
$\phi_n^{(\ge 3)}$ is of degree at least one in $\psi$. Both are of degree
at least three in $(\psi_*,\psi)$.

\item
Set
\begin{align*}
B^{(+)}_{n,\mu,\nu}
   &=\big[D_n^*+Q_{n,\nu}^{(+)}\fQ_nQ_{n,\nu}^{(-)}-\mu\big]^{-1}
     Q_{n,\nu}^{(+)}\fQ_n\\
B^{(-)}_{n,\mu,\nu}
   &=\big[D_n+Q_{n,\nu}^{(+)}\fQ_nQ_{n,\nu}^{(-)}-\mu\big]^{-1}
       Q_{n,\nu}^{(+)}\fQ_n
\end{align*}
where $Q_{n,\nu}^{(+)},Q_{n,\nu}^{(-)}$ were defined in 
\cite[(\eqnPBSqnplusminus)]{POA}.
There are, for each $0\le\nu\le 3$, field maps 
$\phi_{(*)n,\nu}^{(\ge 3)}
  =\phi_{(*)n,\nu}^{(\ge 3)}\big(\psi_*,\psi, \psi_{*\nu},\psi_\nu,\mu,\cV\big)$
such that
\begin{align*}
\partial_\nu\phi_{*n}(\psi_*,\!\psi,\mu,\cV)
&=B^{(+)}_{n,\mu,\nu}\partial_\nu\psi_*
 +\phi_{*n,\nu}^{(\ge 3)}\big(\psi_*,\!\psi, 
                    \partial_\nu\psi_*,\partial_\nu\psi,
                    \mu,\cV\big) \\
\partial_\nu\phi_n(\psi_*,\psi,\mu,\cV)
&=B^{(-)}_{n,\mu,\nu}\partial_\nu\psi
  +\phi_{n,\nu}^{(\ge 3)}\big(\psi_*,\psi, 
                    \partial_\nu\psi_*,\partial_\nu\psi,
                    \mu,\cV\big) 
\end{align*}
and
\begin{equation*}
\TN \phi_{(*)n,\nu}^{(\ge 3)}\TN
\le \GGa_1\,\|V\|_\fm\wf^2\,\wf'
\end{equation*}
Furthermore $\partial_\nu \phi_{(*)n}^{(\ge 3)}(\psi_*,\psi,\mu,\cV)
=\phi_{(*)n,\nu}^{(\ge 3)}\big(\psi_*,\!\psi, 
                    \partial_\nu\psi_*,\partial_\nu\psi,
                    \mu,\cV\big)$, 
and $\phi_{*n,\nu}^{(\ge 3)}$ and $\phi_{n,\nu}^{(\ge 3)}$ are each 
of degree precisely one in $\psi_{(*)\nu}$ and of degree at least two in 
$\big(\psi_*, \psi\big)$.

\item
Set
\begin{align*}
B_{n,\mu,D}^{(+)}
  &=\big[\bbbone-(Q_n^*\fQ_nQ_n-\mu)S_n(\mu)^*\big]Q_n^*\fQ_n\\
B_{n,\mu,D}^{(-)}
  &=\big[\bbbone-(Q_n^*\fQ_nQ_n-\mu)S_n(\mu)\big]Q_n^*\fQ_n
\end{align*}
There are field maps 
$\phi_{(*)n,D}^{(\ge 3)}$
such that
\begin{align*}
D_n^*\phi_{*n}(\psi_*,\!\psi,\mu,\cV)
&=B_{n,\mu,D}^{(+)}\,\psi_*
 +\phi_{*n,D}^{(\ge 3)}\big(\psi_*,\psi,\mu,\cV\big) \\
D_n\phi_n(\psi_*,\psi,\mu,\cV)
&=B_{n,\mu,D}^{(-)}\,\psi
  +\phi_{n,D}^{(\ge 3)}\big(\psi_*,\psi,\mu,\cV\big) 
\end{align*}
and
\begin{equation*}
\TN \phi_{(*)n,D}^{(\ge 3)}\TN
\le \GGa_1\,\|V\|_\fm\wf^3
\end{equation*}
Furthermore $\phi_{(*)n,D}^{(\ge 3)}$ are of degree at least three in 
$\big(\psi_*, \psi\big)$.
\end{enumerate}
\end{proposition}
\begin{proof} (a) 
We shall write the equations \eqref{eqnBGEbgeqns} for 
$\phi_{(*)}(\psi_*,\psi,\mu,\cV)$ in the form
\begin{equation}\label{eqnBGEalgaeqn}
\vec\ga
=\vec f(\vec\al)+\vec L(\vec\al,\vec\ga)+\vec B\big(\vec\al\,;\, \vec\ga\big)
\end{equation}
as in Appendix \ref{appNormsFixedPoint} or in \cite[(\eqnSUBfixedpteqn.b)]{SUB}
with $X=\cX_n$. In particular, we shall use Proposition \ref{SUB:propBGEeqnsoln}
to supply solutions to those equations. Substituting 
\begin{align*}
\al_*&=Q_n^* \fQ_n\psi_*     
&  \al&= Q_n^* \fQ_n\psi  
& \vec\al& = \big(\al_1,\al_2\big)= \big(\al_*,\al\big)\\
\phi_*&=S_n(\mu)^*\big(\al_*+\ga_*\big) 
& \phi&=S_n(\mu)\big(\al+\ga\big) 
& \vec\ga&=\big(\ga_1,\ga_2\big)=\big(\ga_*,\ga\big) 
\end{align*}
into \eqref{eqnBGEbgeqns} gives
\begin{align*}
\ga_* +\cV'_*\big(S_n(\mu)^*(\al_*+\ga_*)\,,\,
                 S_n(\mu)(\al+\ga)\,,\,S_n(\mu)^*(\al_*+\ga_*)\big)
       &=0\\
\ga+\cV'\big(S_n(\mu)(\al+\ga)\,,\,S_n(\mu)^*(\al_*+\ga_*)\,,\,
                  S_n(\mu)(\al+\ga)\big)
       &=0
\end{align*}
We have the desired form with
\begin{align*}
\vec f(\vec\al)(u)
&=\left[\begin{matrix}
   - \cV'_*\big(u;S_n(\mu)^*\al_*,S_n(\mu) \al,S_n(\mu)^*\al_*\big)\\
   - \cV'\big(u;S_n(\mu) \al,S_n(\mu)^*\al_*,S_n(\mu) \al\big)
    \end{matrix}\right]
\\
\noalign{\vskip0.1in}
\vec L(\vec\al;\vec\ga)(u)
&=\left[\begin{matrix}
         - \cV'_*\big(u;S_n(\mu)^*\al_*,S_n(\mu) \ga,S_n(\mu)^*\al_*\big)\hfill\null\\
        \hskip0.5in 
        - 2\cV'_*\big(u;S_n(\mu)^*\al_*,S_n(\mu) \al,S_n(\mu)^*\ga_*\big)
        \\ 
         - \cV'\big(u;S_n(\mu) \al,S_n(\mu)^*\ga_*,S_n(\mu) \al\big)\hfill\null\\
         \hskip0.5in
         - 2\cV'\big(u;S_n(\mu) \al,S_n(\mu)^*\al_*,S_n(\mu) \ga\big)\hfill\null
       \end{matrix}\right]
\\
\noalign{\vskip0.1in}
\vec B(\vec\al;\vec\ga)(u)
&=\left[\begin{matrix} 
          - \cV'_*\big(u;S_n(\mu)^*\ga_*,S_n(\mu) \al,S_n(\mu)^*\ga_*\big)\hfill\null\\
         \hskip0.5in   
         - 2\cV'_*\big(u;S_n(\mu)^*\ga_*,S_n(\mu) \ga,S_n(\mu)^*\al_*\big)\hfill\null\\ 
         \hskip0.5in 
        - \cV'_*\big(u;S_n(\mu)^*\ga_*,S_n(\mu) \ga,S_n(\mu)^*\ga_*\big)\hfill\null\\ 
         - \cV'\big(u;S_n(\mu) \ga,S_n(\mu)^*\al_*,S_n(\mu) \ga\big)\hfill\null\\
          \hskip0.5in 
        - 2\cV'\big(u;S_n(\mu) \ga,S_n(\mu)^*\ga_*,S_n(\mu) \al\big)\hfill\null\\
        \hskip0.5in    
        - \cV'\big(u;S_n(\mu) \ga,S_n(\mu)^*\ga_*,S_n(\mu) \ga\big)\hfill\null
       \end{matrix}\right]
\end{align*}
Here $V(u_1,u_2,u_3,u_4)$ is the kernel of $\cV$ that has the symmetries
\begin{equation}\label{eqnBGEvnperm}
V(u_1,u_2,u_3,u_4)
=V(u_3,u_2,u_1,u_4)
=V(u_1,u_4,u_3,u_2)
\end{equation}
Now apply  \cite[Proposition \propSUBeqnsoln.a\ and Remark \remSUBcB.a]{SUB},
or Proposition \ref{SUB:propBGEeqnsoln},  with $r=s=2$ and
\begin{equation*}
d_{\rm max}=3\qquad
\fc=\half\qquad
\ka_1= \ka_2 = \|Q_n^*\fQ_n\|_{\fm}\,\wf\qquad 
\la_1=\la_2=\wf
\end{equation*}
(and the metric on $X$ being $\fm$ times the metric on $\cX_n$).
Since
\begin{align*}
\tn f_j\tn_w &\le  \|S_n(\mu)\|_{\fm}^3\|V\|_{\fm}\ka_1\ka_2\ka_j\cr
       &\le 8\|S_n\|_{\fm}^3\|Q_n^*\fQ_n\|_{\fm}^3 
              \ \|V\|_{\fm}\wf^3 \\
\tn L_j\tn_{w_{\ka,\la}}
  &\le \|S_n(\mu)\|_{\fm}^3\|V\|_{\fm}\big(2\ka_1\ka_2\la_j
            +\ka_j^2\la_{3-j}\big) \cr
          &\le 24\|S_n\|_{\fm}^3\|Q_n^*\fQ_n\|_{\fm}^2
               \ \|V\|_{\fm}\wf^3\\
\tn B_j\tn_{w_{\ka,\la}}
    &\le \|S_n(\mu)\|_{\fm}^3\ \|V\|_{\fm}\big[  
                    \ka_{3-j}\la^2_j+2\ka_j\la_j\la_{3-j}
                            +\la_j^2\la_{3-j} \big]\cr
    &\le8\|S_n\|_{\fm}^3\big(3\|Q_n^*\fQ_n\|_{\fm}+1\big)\ 
                      \|V\|_{\fm}\wf^3
\end{align*}
assuming that $\rrho_1$ has been chosen small enough that 
$\|S_n(\mu)\|_{\fm}\le 2\|S_n\|_{\fm}$.
By hypothesis, $\tn f_j\tn_w,\,\tn L_j\tn_{w_{\ka,\la}},
\,\tn B_j\tn_{w_{\ka,\la}}<\sfrac{1}{8}\la_j$
and \cite[Proposition \propSUBeqnsoln.a]{SUB} gives a solution $\vec\Ga(\vec\al)$
to \eqref{eqnBGEalgaeqn}
that obeys the bound
\begin{equation*}
\tn \Ga_j\tn_w
\le 16\|S_n\|_{\fm}^3\|Q_n^*\fQ_n\|_{\fm}^3 
              \ \|V\|_{\fm}\wf^3
\end{equation*}
Hence
\begin{alignat*}{3}
\phi_*&=\phi_*(\psi_*,\psi,\mu,\cV)& 
&=S_n(\mu)^*\al_*(\psi_*)
   +S_n(\mu)^*  \Ga_1\big(\al_*(\psi_*),\al(\psi)\big)\\
& &&=S_n(\mu)^*Q_n^*\fQ_n\psi_*
   +S_n(\mu)^*  \Ga_1\big(Q_n^*\fQ_n\psi_*,
                        Q_n^*\fQ_n\psi\big)  \\
\phi&=\phi(\psi_*,\psi,\mu,\cV)& 
&=S_n(\mu)\al(\psi)
      +S_n(\mu)\Ga_2\big(\al_*(\psi_*),\al(\psi)\big)\\
& &&=S_n(\mu)Q_n^*\fQ_n\psi
      +S_n(\mu)\Ga_2\big(Q_n^*\fQ_n\psi_*,
                         Q_n^*\fQ_n\psi\big)
\end{alignat*}
and \cite[Corollary \corSUBsubstitution]{SUB} yields all of the claims.

\Item (b)
We denote  $\phi_{(*)}=\phi_{(*)n}(\psi_*,\psi,\mu,\cV)$. Set
\begin{align*}
S^{(+)}=\big[D_n^*
        +Q_{n,\nu}^{(+)}\fQ_nQ_{n,\nu}^{(-)}-\mu\big]^{-1}\qquad
S^{(-)}=\big[D_n
        +Q_{n,\nu}^{(+)}\fQ_nQ_{n,\nu}^{(-)}-\mu\big]^{-1}
\end{align*}
By \cite[Proposition \:\POGmainpos]{POA}, with $S^{(\pm)}=S_{n,\nu}^{(\pm)}(\mu)$,
we have $\|S^{(\pm)}\|_\fm\le\Gam_\op$, assuming that $\rrho_1$ has been 
chosen small enough.
By \cite[(\eqnPOGpartialS) and Remark \remPBSunderivAlg]{POA},
applying $\partial_\nu$ to \eqref{eqnBGEbgeqns}, and then replacing
$\partial_\nu\phi_{(*)}$ by $\phi_{(*)\nu}$ and 
$\partial_\nu\psi_{(*)}$ by $\psi_{(*)\nu}$ gives
\begin{equation}\label{eqnBGEbgderiv}
\begin{split}
\big(S^{(+)}\big)^{-1}\phi_{*\nu}
   +\cV'_*\big(\phi_{*\nu}\,,\,T_\nu^{-1}\phi\,,
                          \,\phi_*+T_\nu^{-1}\phi_*\big)
          +\cV'_*\big(\phi_*\,,\,\phi_\nu\,,\,\phi_*\big)
     &= Q_{n,\nu}^{(+)} \fQ_n\psi_{*\nu}\\
\big(S^{(-)}\big)^{-1}\phi_\nu
   +\cV'\big(\phi_\nu\,,\,T_\nu^{-1}\phi_*\,,\,\phi+T_\nu^{-1}\phi\big)
          +\cV'\big(\phi\,,\,\phi_{*\nu}\,,\,\phi\big)
   &=Q_{n,\nu}^{(+)} \fQ_n\psi_\nu
\end{split}
\end{equation}
with $T_\nu$ being the translation operator by the lattice basis vector 
in direction $\nu$.
Here we have used the translation invariance of $V$, the symmetries
\eqref{eqnBGEvnperm} and the ``discrete product rule''
\begin{equation}\label{eqnBGEprodrule}
\partial_\nu\big(fg\big)
 =(\partial_\nu f)(T_\nu^{-1} g) +f\partial_\nu g
\end{equation}
in the forms
\begin{equation}\label{eqnBGEprodruleTrip}
\begin{split}
\partial_\nu(fgh)
&= (\partial_\nu f)(T_\nu^{-1} g)(T_\nu^{-1} h) 
              +f(\partial_\nu g)(T_\nu^{-1}h) +fg(\partial_\nu h)\\
\partial_\nu(fgf)
&= (\partial_\nu f)(T_\nu^{-1} g)(T_\nu^{-1} f) 
              +f(T_\nu^{-1}g)(\partial_\nu f) +f(\partial_\nu g)f
\end{split}
\end{equation}
The equations \eqref{eqnBGEbgderiv} are of the form
\begin{equation}\label{eqnBGEbgderivvec}
\vec\ga
=\vec f(\vec\al)+\vec L(\vec\al,\vec\ga)+\vec B\big(\vec\al\,;\, \vec\ga\big)
\end{equation}
as in \cite[(\eqnSUBfixedpteqn.b)]{SUB}, with
\begin{align*}
\al_*&=\phi_*   
     &  \al&= \phi 
     & \al_{*\nu}&=Q_{n,\nu}^{(+)} \fQ_n\psi_{*\nu}     
     &  \al_\nu&= Q_{n,\nu}^{(+)} \fQ_n\psi_\nu   
     & \vec\al& = \big(\al_*,\al,\al_{*\nu},\al_\nu\big)\\
& &
& &
     \phi_{*\nu}&=S^{(+)}(\al_{*\nu}+ \ga_*) &
     \phi_\nu&=S^{(-)}(\al_\nu+\ga) 
    & \vec\ga&=\big(\ga_*,\ga\big) 
\end{align*}
and
\begin{align*}
\vec f(\vec\al)
&=-\left[\begin{matrix}\cV'_*\big(S^{(+)}\al_{*\nu}\,,\,
                          T_\nu^{-1}\al\,,\,\al_*+T_\nu^{-1}\al_*\big) 
                 +\cV'_*\big(\al_*\,,\,S^{(-)}\al_\nu\,,\,\al_*\big) \\
         \cV'\big(S^{(-)}\al_\nu\,,\,T_\nu^{-1}\al_*\,,\,
                                        \al+T_\nu^{-1}\al\big)
        +\cV'\big(\al\,,\,S^{(+)}\al_{*\nu}\,,\,\al\big)
                \end{matrix}\right]
\\
\noalign{\vskip0.1in}
\vec L(\vec\al;\vec\ga)
&=-\left[\begin{matrix}
          \cV'_*\big(S^{(+)}\ga_*\,,\,
                   T_\nu^{-1}\al\,,\,\al_*+T_\nu^{-1}\al_*\big) 
           +\cV'_*\big(\al_*\,,\,S^{(-)}\ga\,,\,\al_*\big) \\
    \cV'\big(S^{(-)}\ga\,,\,T_\nu^{-1}\al_*\,,\,
                                        \al+T_\nu^{-1}\al\big)
    +\cV'\big(\al\,,\,S^{(+)}\ga_*\,,\,\al\big) 
             \end{matrix}\right]
\\
\noalign{\vskip0.1in}
\vec B(\vec\al;\vec\ga)
&=0
\end{align*}
Now apply  \cite[Proposition \propSUBeqnsoln.a]{SUB} with $\fc=\half$ and
\begin{align*}
\ka_1&= \ka_2 = \Gam_\op\|Q_n^*\fQ_n\|_{\fm} \wf
                 +\GGa_1   \|V_n\|_{\fm}\wf^3\\
\la_1&=\la_2=\ka_3=\ka_4=\|Q_{n,\nu}^{(+)}\fQ_n\|_{\fm}\wf'
\end{align*}
Since
\begin{align*}
\tn f_j\tn_w &\le  \max_{\si=+,-}\|S_\mu^{(\si)}\|_{\fm}\,\|V_n\|_\fm
     \Big[ 2 e^{2\veps_n\fm}\ka_1\ka_2\ka_{2+j}+\ka_j^2\ka_{5-j}\Big]
 \le b\la_j \\
\tn L_j\tn_{w_{\ka,\la}}&\le \max_{\si=+,-}\|S_\mu^{(\si)}\|_{\fm}\,\|V_n\|_\fm
     \Big[ 2 e^{2\veps_n\fm}\ka_1\ka_2\la_j+\ka_j^2\la_{3-j}\Big]
 \le b\la_j\\
\tn B_j\tn_{w_{\ka,\la}}&=0 
\end{align*}
where $\veps_n=\sfrac{1}{L^n}$ and 
\begin{align*}
b&=3\max_{\si=+,-}\|S_\mu^{(\si)}\|_{\fm}\ e^{2\veps_n\fm}
         \Big[\Gam_\op\|Q_n^*\fQ_n\|_{\fm} 
             +\GGa_1   \|V\|_{\fm}\wf^2\Big]^2\|V\|_{\fm}\wf^2 
  \le \const\|V\|_{\fm}\wf^2 
  \le\sfrac{1}{4}
\end{align*}
by the hypotheses, \cite[Propositions \propSUBeqnsoln.a]{SUB} gives
a solution $\vec\Ga(\vec\al)$ to \eqref{eqnBGEbgderivvec} with 
\begin{align*}
&\tn \Ga_1\tn_{w_{\ka,\la}},
\tn \Ga_2\tn_{w_{\ka,\la}}
\le \GGa_1' \|V\|_{\fm}\wf^2 \wf'
\end{align*}
As \eqref{eqnBGEbgderivvec} is a linear system of equations and 
$b\le\sfrac{1}{4}$, the solution is unique. Correspondingly 
\begin{align*}
\phi_{*\nu}
  &= B^{(+)}_{n,\mu,\nu}\psi_{*\nu}
    +S^{(+)}  \Ga_1\big(\al_*(\phi_*),\al(\phi),
               \al_{*\nu}(\psi_{*\nu}),\al_\nu(\psi_\nu)\big)\\
\phi_\nu
  &= B^{(-)}_{n,\mu,\nu}\psi_\nu
     +S^{(-)}\Ga_2 \big(\al_*(\phi_*),\al(\phi),
               \al_{*\nu}(\psi_{*\nu}),\al_\nu(\psi_\nu)\big)
\end{align*}
solves \eqref{eqnBGEbgderiv}.
The conclusion now follows by part (a) and 
\cite[Corollary \corSUBsubstitution]{SUB}. 

That $\partial_\nu \phi_{(*)n}^{(\ge 3)}(\psi_*,\psi,\mu,\cV)
=\phi_{(*)n,\nu}^{(\ge 3)}\big(\psi_*,\!\psi, 
                    \partial_\nu\psi_*,\partial_\nu\psi,
                    \mu,\cV\big)$
follows from the observation that $\partial_\nu S_n(\mu)^{(*)}Q_n^* \fQ_n
=B^{(\pm)}_{n,\mu,\nu}$, by
\cite[(\eqnPOGpartialS) and Remark \remPBSunderivAlg]{POA}.

\Item (c)
From \eqref{eqnBGEbgeqns} we see  
\begin{align*}
   D_n^*\phi_*
    &= Q_n^*\fQ_n\psi_*
       -\big(Q_n^*\fQ_nQ_n-\mu\big)\phi_*
       -\cV'_*(\phi_*,\phi,\phi_*)\\
 D_n\phi
   &=Q_n^*\fQ_n\psi
       -\big(Q_n^*\fQ_nQ_n-\mu\big)\phi
       -\cV'(\phi,\phi_*,\phi)
\end{align*}
with $\phi_{(*)}=\phi_{(*)n}$. 
Now just substitute for $\phi_{(*)n}$ using part (a). 
\end{proof}

\begin{remark}[The complex conjugate of the background field]\label{remBGEbgfandcomplexconj}
Assume that the constants $\GGa_1$, $\rrho_1>0$ of Proposition \ref{propBGEphivepssoln}
are chosen big enough and small enough, respectively, and 
fulfil its hypotheses.
Let $\psi(x)$ be a field on $\cX_0^{(n)}$ such that $|\psi(x)| <\wf$ and
$|\partial_\nu \psi(x)| <\wf'$ for all $x\in \cX_0^{(n)}$ and $0\le \nu
\le 3$. Then
\begin{equation*}
\big|   \phi_{*n}(\psi^*,\psi,\mu,\cV)^*(u)- \phi_{n}(\psi^*,\psi,\mu,\cV)(u)\big|\le
\GGa_1 \wf'   \qquad \qquad {\rm for\ all\ } u \in \cX_n
\end{equation*}
\end{remark}

\begin{proof}
Write $\phi_{(*)} = \phi_{(*)n}(\psi^*,\psi,\mu,\cV)$.
By Proposition \ref{propBGEphivepssoln} and \cite[Lemma \lemSUBLp.b]{SUB}
\begin{equation}\label{eqnBGEsupphiderphi}
|\phi(u)| \le \GGa_1 \wf  \quad {\rm and} \quad |\partial_\nu\phi(u)| \le \GGa_1 \wf'
\qquad {\rm for \ all\ } u \in \cX_n,\ 0\le\nu\le 3
\end{equation}
By \eqref{eqnBGEbgeqns} and the fact that 
$\,S^{-1}_n(\mu)  -S^{-1}_n(\mu)^\dagger = D_n-D_n^\dagger\,$ (see 
the definition of $S_n(\mu)$ after \eqref{eqnBGEbgeqns})
\begin{equation*}
S^{-1}_n(\mu)( \phi_*^*- \phi)+\cV'_*(\phi_*,\phi,\phi_*)^*-\cV'(\phi,\phi_*,\phi)
              =(D_n-D_n^\dagger) \phi_*^*   
\end{equation*}
where $\dagger$ refers to the adjoint.
Localizing as in \cite[Corollary \corLprelocalize]{PAR2},
\begin{equation}\label{eqnBGEestcalVL}
S^{-1}_n(\mu)( \phi_*^*- \phi)
+ \rv \phi^* ( \phi_*^*+ \phi)\,( \phi_*^*- \phi)  
- \rv \phi^2\,( \phi_*^*- \phi)^*  
              =(D_n-D_n^\dagger) \,\phi_*^* + \cV_\loc(\phi_*,\phi)
\end{equation}
where $\rv = \int V(0,u_1,u_2,u_3)\ du_1\,du_2\,du_3$ and $ \cV_\loc(\phi_*,\phi)$ is a field such that
\begin{equation*}
\big|  \cV_\loc(\phi_*,\phi)(u)\big| \le \const \wf'  \qquad \qquad {\rm for\ all\ } u \in \cX_n
\end{equation*}
By \cite[(\eqnPDOdndef)]{POA},
\begin{align*}
D_n-D_n^\dagger
&=L^{2n}\ \bbbl_*^{-n} e^{-\oh_0}
         \big( \partial_0^\dagger- \partial_0\big)\bbbl_*^n \\
&=e^{-\bbbl_*^{-n}\oh_0\bbbl_*^n}
        \ \big( \partial_0^\dagger- \partial_0\big) 
\end{align*}
Beware that in the first line $\partial_0$ acts on the $\cH_0^{(n)}$,
while in the second line $\partial_0$ acts on $\cH_n$.
Hence, by \eqref{eqnBGEsupphiderphi}
\begin{equation}\label{eqnBGEestDnDndagger}
\big| (D_n-D_n^\dagger) \phi_*^*(u) \big| \le \const \wf'
 \qquad \qquad {\rm for\ all\ } u \in \cX_n
\end{equation}
Also considering the complex conjugate, we see that $\si =  \phi_*^*- \phi$
fulfils the equations
\begin{equation}\label{eqnBGEeqnfordifference}
\begin{split}
\big[\bbbone +S_n(\mu)\,\rv  \phi^* ( \phi_*^*+\phi) \big] \si 
-  S_n(\mu) \,\rv \phi^2\,\si^*
     &=S_n(\mu)\big[(D_n-D_n^\dagger) \,\phi_*^* + \cV_\loc(\phi_*,\phi)\big]
\\
\big[\bbbone +\overline{S_n(\mu)}\,\rv  \phi ( \phi_*+ \phi^*) \big]\si^*
-  \overline{S_n(\mu)} \,\rv {\phi^*}^2\,\si 
     &= \overline{S_n(\mu)}\,\big[ (\overline{D_n}-D_n^*) \,\phi_*
                      + \cV_\loc(\phi_*,\phi)^*\,\big]
\end{split}
\end{equation}
where, in the square brackets on the left hand side, 
$\phi^* ( \phi_*^*+ \phi) $ and $\phi ( \phi_*+\phi^*)$, respectively,
are viewed as multiplication operators. 
By \cite[Proposition \POGmainpos]{POA} and \eqref{eqnBGEsupphiderphi},
the $L^1$--$L^\infty$ norm of the operators 
$S_n(\mu)\,\rv  \phi^* ( \phi_*^*+\phi) $
and $ S_n(\mu) \,\rv \phi^2$ is bounded by $2\GGa_\op\,\rrho_1\le\sfrac{1}{4}$.
Hence, one can solve \eqref{eqnBGEeqnfordifference} for $\si$ and $\si^*$, 
and the estimates \eqref{eqnBGEestcalVL} and
\eqref{eqnBGEestDnDndagger} for the terms on the right hand side give the desired estimate.
\end{proof}

\begin{remark}[Third order terms of the background field] \label{remBGEphivepssoln}
Proposition  \ref{propBGEphivepssoln}.a states that the linear part 
of the background field $\phi_{(*)n}(\psi_*,\psi,\mu,\cV)$ is
\begin{equation*}
\phi_{(*)n}^{(1)}(\psi_*,\psi,\mu,\cV) = S_n(\mu)^{(*)}Q_n^* \fQ_n\,\psi_{(*)}
\end{equation*}
and that the higher order terms $\phi_{(*)n}^{(\ge 3)}$ are of degree at
least three in $\psi_*,\psi$. In fact, the term of degree exactly three
can be described easily. There is a constant $\hat \GGa_1$ and there are field maps  $\phi_{(*)n}^{(\ge 5)}$ such that
\begin{align*}
\phi_{*n}^{(\ge 3)}(\psi_*,\psi,\mu,\cV)
  &= - S_n(\mu)\,\cV_*'\big(\Phi_*,\Phi,\Phi_*\big) 
\big|_{\Phi_{(*)}= \phi_{(*)n}^{(1)}(\psi_*,\psi,\mu,\cV) }
        \ \   + \phi_{(*)n}^{(\ge 5)}(\psi_*,\psi,\mu,\cV)
\\
\phi_{n}^{(\ge 3)}(\psi_*,\psi,\mu,\cV)
  &= - S_n(\mu)\,\cV'\big(\Phi,\Phi_*,\Phi\big) 
\big|_{\Phi_{(*)}= \phi_{(*)n}^{(1)}(\psi_*,\psi,\mu,\cV) }
        \ \   + \phi_{(*)n}^{(\ge 5)}(\psi_*,\psi,\mu,\cV)
\end{align*}
and
$\ 
\TN \phi_{(*)n}^{(\ge 5)}\TN
\le \hat\GGa_1 \|V\|_{\fm}^2\wf^5
\ $.
\end{remark}
\begin{proof}  We prove the statement about $\phi_n^{(\ge 3)}$.
Write $\phi_{(*)}=\phi_{(*)n}(\psi^*,\psi,\mu,\cV)$  and 
$\Phi_{(*)}= \phi_{(*)n}^{(1)}(\psi_*,\psi,\mu,\cV)$.
By \eqref{eqnBGEbgeqns},
\begin{align*}
\phi &=S_n(\mu) Q_n^* \fQ_n\psi - S_n(\mu)\cV'(\phi,\phi_*,\phi)\\
  &=\Phi -  S_n(\mu)\cV'(\Phi+\phi^{(\ge 3)},
                                 \Phi_*+\phi_*^{(\ge 3)},
                                 \Phi+\phi^{(\ge 3)})\\
&=\Phi- S_n(\mu)\,\cV'\big(\Phi,\Phi_*,\Phi\big) 
          + \phi_{(*)n}^{(\ge 5)}(\psi_*,\psi,\mu,\cV)
\end{align*}
with
\begin{equation*}
 \phi_n^{(\ge 5)}(\psi^*,\psi)
   = - S_n(\mu)\big\{\cV'(\Phi+\phi^{(\ge 3)},
                                 \Phi_*+\phi_*^{(\ge 3)},
                                 \Phi_*+\phi_*^{(\ge 3)})
                  - \cV'(\Phi,\Phi_*,\Phi)\big\}
\end{equation*}
The estimate on $ \phi_n^{(\ge 5)}$ follows from Proposition  \ref{propBGEphivepssoln}.a
and \cite[Lemma \lemSUBdiff]{SUB}.
\end{proof}

To derive a representation of the background fields of the form 
\eqref{eqnBGEbglocal} from Proposition \ref{propBGEphivepssoln}, we use
\begin{lemma}\label{lemBGEexpandalphiveps}  
There are field maps 
     $F_\lb\big(\{\psi_\nu\}\big)$ 
and 
     $F_{\lb*}\big(\{\psi_{*\nu}\}\big)$ 
and a constant $\GGa_2$ such that
\begin{align*}
\big(S_n(\mu)^{(*)}Q_n^* \fQ_n\psi_{(*)}\big)(u)
&=\sfrac{a_n}{a_n-\mu}\psi_{(*)}\big(X(u)\big)
          +F_{\lb(*)}(\{\partial_\nu\psi_{(*)}\})(u)
\end{align*}
and
\begin{align*}
\tn F_{\lb(*)}\tn&\le  \GGa_2 \,\wf'\,
   \|S_n(\mu)^{(*)}Q_n^* \fQ_n\|_{\bar\fm}
\end{align*}
Furthermore, the maps $F_{\lb(*)}$ are of degree precisely one.

\end{lemma}
\begin{proof}  We prove the lemma for
$
B=S_n(\mu)Q_n^* \fQ_n
$.
Denote by $1$ and $1_\fin$ the constant fields on $\cX_0^{(n)}$ and
$\cX_n$, respectively, that always take the value $1$.
By \cite[Remark \remMXconsteigen]{PAR2}, $Q_n1_\fin=1$, $Q_n^*1=1_\fin$ 
and $\fQ_n 1=a_n1$. Since $D_n$ annihilates constant fields,
\begin{align*}
B\,1
=S_n(\mu)Q_n^* \fQ_n 1
=(D_n+Q_n^*\fQ_nQ_n-\mu)^{-1}Q_n^* \fQ_n 1
=\sfrac{a_n}{a_n-\mu}1_\fin
\end{align*}
Fix any $u\in\cX_n$ and any field $\psi$ on $\cX_0^{(n)}$.  Then
\begin{align*}
&\big(B\psi\big)(u)
=\sum_{x\in\cX^{(n)}_0} B(u,x)\,\psi(x)\\
&\hskip0.5in
=\sum_{x\in\cX^{(n)}_0} B(u,x)\,
                              \psi\big(X(u)\big)
+\sum_{x\in\cX^{(n)}_0} 
    B(u,x)\,
    \big[\psi(x)-\psi\big(X(u)\big)\big]\\
&\hskip0.5in
=\sfrac{a_n}{a_n-\mu}\psi\big(X(u)\big)
+\sum_{x\in\cX^{(n)}_0} 
   B(u,x)\, \big[\psi(x)-\psi\big(X(u)\big)\big]
\end{align*}
It now suffices to apply \cite[Lemma \lemLexpandalphiveps]{PAR2}. 
\end{proof}

\begin{corollary}\label{corBGEexpandaldephi}
There are field maps  $\brphi_{(*)n}$ 
and a constant $\GGa_3$ such that, under 
the hypotheses of Proposition \ref{propBGEphivepssoln},
\begin{align*}
\phi_{(*)n}(\psi_*,\psi,\mu,\cV)(u)
&=\sfrac{a_n}{a_n-\mu}\psi_{(*)}\big(X(u)\big)
           +\brphi_{(*)n}\big((\psi_*,\{\partial_\nu\psi_*\})\,,\,
                       (\psi,\{\partial_\nu\psi\})\,,\,\mu,\cV\big)(u)
\end{align*}
and
$$
\tn \brphi_{(*)n} \tn
\le \GGa_3\big(\wf' + \|V\|_{\fm}\wf^3\big)
$$
\end{corollary}
\begin{proof}  Proposition \ref{propBGEphivepssoln}.a and Lemma \ref{lemBGEexpandalphiveps} imply that
\begin{align*}
\phi_{(*)n}(\psi_*,\psi,\mu,\cV)(u)
&=\big(S_n(\mu)^{(*)}Q_n^* \fQ_n\psi_{(*)}\big)(u)
      +\phi_{(*)n}^{(\ge 3)}(\psi_*,\psi,\mu,\cV)(u)\\
&=\sfrac{a_n}{a_n-\mu}\psi_{(*)}\big(X(u)\big)
           +F_{\lb(*)}(\{\partial_\nu\psi_{(*)}\})(u)
      + \phi_{(*)n}^{(\ge 3)}(\psi_*,\psi,\mu,\cV)(u)\\
&=\sfrac{a_n}{a_n-\mu}\psi_{(*)}\big(X(u)\big)
           +\brphi_{(*)n}\big((\psi_*,\{\partial_\nu\psi_*\})\,,\,
                       (\psi,\{\partial_\nu\psi\})\,,\,\mu,\cV\big)(u)
\end{align*}
with
$$
\tn \brphi_{(*)n}\tn
\le \GGa_2 
   \|S_n(\mu)^{(*)}Q_n^* \fQ_n\|_{\bar\fm}\ \wf'
   +\GGa_1 \|V\|_{\fm}\wf^3
\le  \GGa_3\big(\wf' + \|V\|_{\fm}\wf^3\big)
$$
\end{proof}

\newpage
\section{Variations of the Background Field with Respect to $\psi$}
\label{secBGEvrnspsi}

Recall from \cite[(\eqnSTdehatphidef)]{PAR1} that 
\begin{equation}\label{eqnBGEdehatphidef}
\de\hat\phi_{(*)n+1}(\psi_*,\psi,z_*,z)
  =\bbbs\big[\de\check\phi_{(*)n+1}\big(\bbbs^{-1}\psi_*\,,\,
      \bbbs^{-1}\psi_*\,,\,
      D^{(n)*}\bbbl_* z_*\,,\,
      D^{(n)}\bbbl_* z\,,\,
      \mu,\cV\big)\big]
\end{equation}
where
\begin{itemize}[leftmargin=*, topsep=2pt, itemsep=2pt, parsep=0pt]
\item[$\circ$] the fields 
\begin{align*}
&\de\check\phi_{(*)n+1}(\th_*,\th,\de\psi_*,\de\psi,\mu,\cV) \\ &\hskip0.5in
= \Big[\phi_{(*)n}\big(\psi_{*}+\de\psi_*,\psi+\de\psi,\mu,\cV\big)
       - \phi_{(*)n}\big(\psi_{*},\psi,\mu,\cV\big)\Big]_{\psi_{(*)}=\psi_{(*)n}(\th_*,\th,\mu,\cV)}
\end{align*}
were defined in \cite[Definition \defBGAbckgndVarn.a]{PAR1},
\item[$\circ$] 
the scaling operators $\bbbs$ and $\bbbl_*$ were defined in 
\cite[Appendix \appDEFscaling]{PAR1}, and
\item[$\circ$] the  operator  square root $\,D^{(n)}\,$ of the fluctuation 
field covariance $C^{(n)}$ was defined just before \cite[(\eqnHTcn)]{PAR1}.

\end{itemize}
The fields $\de\hat\phi_{(*)n+1}$ also depend implicitly on $\mu$ and $\cV$.
Proposition \ref{propBGEdephisoln}, below, implies that $\de\hat\phi_{(*)n+1}$
are analytic maps in $(\psi_*,\psi,z_*,z)$ from a neighborhood 
of the origin in
$\,\cH_0^{(n+1)}\times \cH_0^{(n+1)}\times \cH_1^{(n)}\times \cH_1^{(n)}\,$ 
to $\,\cH_{n+1}^{(0)}$. 
As in \cite[\S\chapOSFfluct]{PAR2},  we define, on the space of field maps
$F(\psi_*,\psi,z_*,z)$, the projections
\begin{itemize}[leftmargin=*, topsep=2pt, itemsep=2pt, parsep=0pt]
\item[$\circ$] 
$P^\psi_2$ 
which extracts the part of degree exactly one in each of $\psi_*$ 
and $\psi$, and of arbitrary degree in $z_{(*)}$ and
\item[$\circ$]
$P^\psi_1$ 
which extracts the part of degree exactly one in $\psi_{(*)}$, 
and of arbitrary degree in $z_{(*)}$  and
\item[$\circ$]
$P^\psi_0$  which extracts the part of degree zero in 
$\psi_{(*)}$ and of arbitrary degree in $z_{(*)}$. 
\end{itemize}

\begin{proposition}\label{propBGEdephisoln}
There are constants\footnote{Recall Convention \ref{convBGEconstants}.} 
$\GGa_4$ and $\rrho_2>0$ such that the following hold, if
\begin{equation*}
\max\big\{\ L^2|\mu|\ ,\  
       \|V\|_\fm (\wf+L^9\wf_\fl)(\wf+\wf'+L^9\wf_\fl)\ \big\}
\le \rrho_2
\end{equation*}
\begin{enumerate}[label=(\alph*), leftmargin=*]
\item[$\circ$]
The field maps $\de\hat\phi_{(*)n+1}(\psi_*,\psi,z_*,z)$ obey\ \ \ 
$
\tn\de\hat\phi_{(*)n+1}\tn
\le L^{11}\GGa_4\, \wf_\fl
$.

\item[$\circ$]
Write, as in \cite[(\eqnOSAhatphiplus)]{PAR1}
\begin{align*}
\de\hat\phi_{(*)n+1}^{(+)}(\psi_*,\psi,z_*,z)
=\de\hat\phi_{(*)n+1}(\psi_*,\psi,z_*,z)
    -L^{3/2}\bbbs S_n^{(*)}Q_n^*\fQ_n D^{(n)(*)}\bbbs^{-1}z_{(*)}
\end{align*}
It obeys\ \ \ 
$
\TN \de\hat\phi_{(*)n+1}^{(+)}\TN
\le L^{29} \GGa_4\,\{\|V\|_\fm(\wf+\wf_\fl)^2+|\mu|\}
              \wf_\fl
$.
\item[$\circ$]
The part, $\de\hat\phi_{(*)n+1}^{(\ge 2)}$, of 
$\de\hat\phi_{(*)n+1}^{(+)}$ that is of degree 
at least two in $z_{(*)}$, fulfils the bound
\begin{align*}
\TN \de\hat\phi_{(*)n+1}^{(\ge 2)}\TN
& \le L^{29} \GGa_4\, \|V\|_\fm(\wf+\wf_\fl)\wf_\fl^2
\end{align*}

\item[$\circ$] 
Using the notation of \cite[Definition \defBGAgradV]{PAR1}, we have
\begin{align*}
\de\hat\phi_{*n+1}^{(+)}(\psi_*,\psi,z_*,z)
&=L^{3/2}\bbbs [S_n(\mu)^*-S_n^*]Q_n^*\fQ_n 
                                   D^{(n)*}\bbbs^{-1}z_*\\
&\hskip0.75in
  - L^{\sfrac{3}{2}}\bbbl_*^{-1} S_n(\mu)^*\cV'_*(\varphi_*, \varphi,\varphi_*) 
      \Big|^{\atop{\varphi_*=\phi_*+\de\phi_*}
                  {\varphi=\phi +\de\phi}}
     _{\atop{\varphi_*=\phi_*}
            {\varphi=\phi}}
  + \de\hat\phi_*^{(\ho)}\cr
\de\hat\phi_{n+1}^{(+)}(\psi_*,\psi,z_*,z)
&=L^{3/2}\bbbs [S_n(\mu)-S_n]Q_n^*\fQ_n 
                                   D^{(n)}\bbbs^{-1}z\\
&\hskip0.75in
  - L^{\sfrac{3}{2}}\bbbl_*^{-1}S_n(\mu)\cV'(\varphi, \varphi_*,\varphi) 
      \Big|^{\atop{\varphi_*=\phi_*+\de\phi_*}
                  {\varphi=\phi +\de\phi}}
           _{\atop{\varphi_*=\phi_*}
                  {\varphi=\phi}}
  + \de\hat\phi^{(\ho)}
\end{align*}
with the substitutions
\begin{align*}
\phi_{(*)}
    &=\bbbs^{-1}S_{n+1}(L^2\mu)^{(*)}Q_{n+1}^* \fQ_{n+1}\,\psi_{(*)}\\
\de\phi_{(*)}
   &=S_n(\mu)^{(*)} Q_n^* \fQ_n\,L^{3/2} D^{(n)(*)}\bbbs^{-1} z_{(*)}
\end{align*}
and with the contributions in $\de\hat\phi_{(*)}^{(\ho)}$ being 
of degree at least five in
  $(\psi_{(*)},z_{(*)})$ and obeying
\begin{align*}
  \TN P^\psi_j\de\hat\phi_{(*)}^{(\ho)} \TN   
        &\le  L^{(1- j)3/2}L^{9(5-j)}
                   \GGa_4 \|V\|_\fm^2\wf^j\wf_\fl^{5-j} \qquad
   \text{for $j=0,1,2$}
\end{align*}

\item[$\circ$]
 There are field maps 
  $\de\hat\phi_{(*)n+1,\nu}\big(\psi_*,\psi,\psi_{*\nu},\psi_\nu,z_*,z\big)$,
   $0\le\nu\le 3$, 
such that
\begin{align*}
&\big(\partial_\nu\de\hat\phi_{(*)n+1}\big)(\psi_*,\psi,z_*,z)
 =\de\hat\phi_{(*)n+1,\nu}\big(\psi_*,\psi,\partial_\nu\psi_*,\partial_\nu\psi,
                      z_*,z\big)
\end{align*}
and
\begin{equation*}
\tn \de\hat\phi_{(*)n+1,\nu}\tn
\le L_\nu L^{11}\GGa_4\,\wf_\fl 
\end{equation*}
where $L_0=L^2$ and $L_\nu=L$ for $\nu=1,2,3$.
\end{enumerate}
\end{proposition}

\noindent 
This Proposition will be proven following the proof of 
Lemma \ref{lemBGEdephisolnB}.
Recall, from \eqref{eqnBGEdehatphidef}, that $\de\hat\phi_{(*)n+1}$
is defined in terms of $\de\check\phi_{(*)n+1}$.
Also recall, from \cite[Remark \remBGAbckgndVarn.c and 
Definition \defBGAgradV]{PAR1},
that 
$\de\check\phi_{(*)n+1}$ is obtained from the solution 
$\de\phi_{(*)}\!=\de\varphi_{(*)}(\phi_*,\phi,\de\psi_*,\de\psi)$ of
\begin{equation}\label{eqnBGEdevarphi}
\begin{split}
 \de\phi_*
&={S_n^*}Q_n^* \fQ_n\, \de\psi_*  +\mu {S_n^*}\de\phi_*
  - {S_n^*}\cV'(\varphi_*, \varphi,\varphi_*) 
      \Big|^{\varphi_*=\phi_*+\de\phi_*\atop
        \varphi=\phi+\de\phi}
     _{\varphi_*=\phi\atop
        \varphi=\phi}
 \\
\noalign{\vskip0.05in}
\de\phi
&= S_n Q_n^* \fQ_n\, \de\psi +\mu S_n\de\phi
- S_n\cV'_*(\varphi,\varphi_*, \varphi) 
      \Big|^{\varphi_*=\phi_*+\de\phi_*\atop
        \varphi=\phi+\de\phi}
     _{\varphi_*=\phi_*\atop
        \varphi=\phi}
\end{split}
\end{equation}
by substituting $\phi_{(*)}=\check\phi_{(*)n+1}(\th_*, \th,\mu,\cV)$. 
So we first prove the existence of and develop bounds on 
$\de\varphi_{(*)}(\phi_*,\phi,\de\psi_*,\de\psi)$.
We fix any $\wf_\phi$,$\wf'_\phi$, $\wf_{\de\psi}\ge 1$ and
denote by $\tn\ \cdot\ \tn_\phi$ the (auxiliary) norm with 
mass $\fm$ that assigns the weight factors
$\wf_\phi$ to the fields $\phi_{(*)}$,
$\wf'_\phi$ to the fields $\phi_{(*)\nu}$ and
$\wf_{\de\psi}$ to the fields $\de\psi_{(*)}$.

\begin{lemma}\label{lemBGEdephisoln}
There are constants $\GGa'_4$ and $\rrho'_2>0$ such that the following hold, 
if
$$
\max\big\{\ |\mu|\ ,\  
 \|V\|_\fm (\wf_\phi+\wf_{\de\psi})(\wf_\phi+\wf'_\phi+\wf_{\de\psi})\ \big\}
\le \rrho'_2
$$
\begin{enumerate}[label=(\alph*), leftmargin=*]
\item[$\circ$]
There are field maps $\de\varphi_{(*)}(\phi_*,\phi,\de\psi_*,\de\psi)$
that obey
$
\TN \de\varphi_{(*)}\TN_\phi
\le \GGa'_4\, \wf_{\de\psi}
$ and solve \eqref{eqnBGEdevarphi}.
Write
\begin{align*}
\de\varphi_{(*)}&= S_n^{(*)}Q_n^*\fQ_n\de\psi_{(*)}
    +\de\varphi_{(*)}^{(+)}
\end{align*}
and denote by  $\de\varphi_{(*)}^{(\ge 2)}$ the part of 
$\de\varphi_{(*)}^{(+)}$ that is of degree at least two in 
$\de\psi_{(*)}$. They obey
\begin{align*}
\TN \de\varphi_{(*)}^{(+)}\TN_\phi
&\le \GGa'_4\,\{\|V\|_\fm(\wf_\phi+\wf_{\de\psi})^2+|\mu|\}
              \wf_{\de\psi}\\
\TN \de\varphi_{(*)}^{(\ge 2)}\TN_\phi
&\le \GGa'_4\, \|V\|_\fm(\wf_\phi+\wf_{\de\psi})\wf_{\de\psi}^2
\end{align*}

\item[$\circ$]
There are field maps 
   $\de\varphi_{(*)\nu}\big(\phi_*,\phi,\phi_{*\nu},\phi_\nu,
                   \de\psi_*,\de\psi\big)$,
   $0\le\nu\le 3$, 
such that
\begin{align*}
&\big(\partial_\nu\de\varphi_{(*)}\big)(\phi_*,\phi,\de\psi_*,\de\psi)
    =\de\varphi_{(*)\nu}\big(\phi_*,\phi,\partial_\nu\phi_*,\partial_\nu\phi,
                      \de\psi_*,\de\psi\big)
\end{align*}
and\ \ \ 
$
\tn \de\varphi_{(*)\nu}\tn_\phi
\le \GGa'_4\,\wf_{\de\psi} 
$.
\end{enumerate}
\end{lemma}
\begin{proof} (a) The equations \eqref{eqnBGEdevarphi}, for 
$\de\varphi_{(*)}$, are of the form
\begin{equation*}
\vec\ga
=\vec f(\vec\al)+\vec L(\vec\al,\vec\ga)+\vec B\big(\vec\al\,;\, \vec\ga\big)
\end{equation*}
as in \cite[(\eqnSUBfixedpteqn.b)]{SUB}, with $X=\cX_n$ and
\begin{alignat*}{5}
\al_*&=\phi_* \quad  
     &  \al&= \phi \quad
     & \de\al_*&=  Q_n^*\fQ_n\de\psi_* \quad   
     &  \de\al&= Q_n^*\fQ_n\de\psi  \quad
     & \vec\al& = \big(\al_*,\al,\de\al_*,\de\al\big)\cr
& &
& &
     \de\phi_*&= S_n^* \ga_* &
     \de\phi&= S_n\ga
    & \vec\ga&=\big(\ga_*,\ga\big) \cr 
\end{alignat*}
and
\begin{align*}
\vec f(\vec\al)(u)
&=\left[\begin{matrix}
                \de\al_*(u)  \\
                \de\al(u)
        \end{matrix}\right]
\cr
\noalign{\vskip0.1in}
\vec L(\vec\al;\vec\ga)(u)
&=\left[\begin{matrix}
        \mu (S_n^*\ga_*)(u) 
                     - \cV'_*\big(u;\al_*,S_n \ga,\al_*\big)
                     - 2\cV'_*\big(u;\al_*,\al,S_n^*\ga_*\big)  \\
           \noalign{\vskip0.05in}
          \mu (S_n\ga)(u) 
               - \cV'\big(u;\al,S_n^*\ga_*,\al\big)
               - 2\cV'\big(u;\al,\al_*,S_n \ga\big)
       \end{matrix}\right]
\\
\noalign{\vskip0.1in}
\vec B(\vec\al;\vec\ga)(u)
&=\left[\begin{matrix}
          - \cV'_*\big(u;S_n^*\ga_*,\al,S_n^*\ga_*\big)
          - 2\cV'_*\big(u;S_n^*\ga_*,S_n \ga,\al_*\big)\hfill\null\\
         \hskip2.5in - \cV'_*\big(u;S_n^*\ga_*,S_n \ga,S_n^*\ga_*\big)\\
          \noalign{\vskip0.05in}
               - \cV'\big(u;S_n \ga,\al_*,S_n \ga\big)
               - 2\cV'\big(u;S_n \ga,S_n^*\ga_*,\al\big)\hfill\null\\
         \hskip2.5in- \cV'\big(u;S_n \ga,S_n^*\ga_*,S_n \ga\big)
        \end{matrix}\right]
\end{align*}
Now apply  \cite[Proposition \propSUBeqnsoln.a and Remark \remSUBcB.a]{SUB} 
with $d_{\rm max}=3$, $\fc=\half$ and 
\begin{align*}
&\ka_1= \ka_2 = \wf_\phi &
\ka_3&=\ka_4=\| Q_n^*\fQ_n\|_\fm\wf_{\de\psi} &
&\la_1=\la_2=4\ka_4 &
\end{align*}
Since
\begin{align*}
\tn f_j\tn_w&\le  \ka_{2+j}
               =\sfrac{1}{4}\la_j\\
\tn L_j\tn_{w_{\ka,\la}}
   &\le \|S_n\|_{\fm}(|\mu|+2\|V\|_{\fm}\ka_1\ka_2)\la_j
            +\|S_n\|_{\fm}\ \|V\|_{\fm}\ka_j^2\la_{3-j}\\
  &\le  \|S_n\|_\fm \Big\{|\mu|+3\|V\|_\fm\wf_\phi^2\big\}\la_j\\
\tn B_j\tn_{w_{\ka,\la}}
    &\le\|S_n\|_{\fm}^2\ \|V\|_{\fm}\big[  
                    \ka_{3-j}\la^2_j+2\ka_j\la_j\la_{3-j}
                            +\|S_n\|_{\fm}\la_j^2\la_{3-j} \big] \\
  &\le \|S_n\|_\fm^2\|Q_n^*\fQ_n\|_\fm\|V\|_\fm
      \Big\{12\wf_\phi \wf_{\de\psi}
  +16\|S_n\|_\fm\|Q_n^*\fQ_n\|_\fm \wf_{\de\psi}^2
      \Big\}\la_j 
\end{align*}
\cite[Proposition \propSUBeqnsoln.a]{SUB} gives
\begin{equation*}
\de\varphi_{(*)}(\phi_*,\phi,\de\psi_*,\de\psi)
  =\de\phi_{(*)}
  = S_n^{(*)}\Ga_{(*)}\big(\phi_*,\phi, Q_n^*\fQ_n\de\psi_*,
                                        Q_n^*\fQ_n\de\psi\big)
\end{equation*}
with 
$
\tn \Ga_{(*)}\tn_{w_{\ka,\la}}
       \le 2\| Q_n^*\fQ_n\|_\fm \wf_{\de\psi}
$.
The first conclusion now follows. 

Denote by $\de\varphi_{(*)}^{(1)}$ the part of $\de\varphi_{(*)}$ 
that is of degree precisely one in $\de\psi_{(*)}$ and decompose
\begin{align*}
\de\varphi_{(*)}^{(1)}&= S_n^{(*)}Q_n^*\fQ_n\de\psi_{(*)}
    +\de\varphi_{(*)}^{(1)\squig}
\end{align*}
In the notation of \cite[Proposition \propSUBeqnsoln.b]{SUB},
$\vec\Ga^{(1)}$ is the part of $\vec \Ga$ that is of degree
precisely $1$ in $\vec f$. In our application, $\vec f$ is homogeneous
of degree one in $\de\psi_{(*)}$, and $\de\psi_{(*)}$ does not appear 
in either $\vec L$ or $\vec B$, so
\begin{align*}
\de\varphi_{(*)}^{(1)} &= S_n^{(*)}\Ga_{(*)}^{(1)}
     \big(\phi_*,\phi, Q_n^*\fQ_n\de\psi_*,Q_n^*\fQ_n\de\psi\big) \\
\de\varphi_{(*)}^{(1)\squig} &=
 S_n^{(*)}\big\{ \Ga_{(*)}^{(1)}
     \big(\phi_*,\phi, Q_n^*\fQ_n\de\psi_*,Q_n^*\fQ_n\de\psi\big)
     -f_{(*)}\big(Q_n^*\fQ_n\de\psi_*,Q_n^*\fQ_n\de\psi\big)\big\}\\
\de\varphi_{(*)}^{(\ge 2)} &=
 S_n^{(*)}\big\{ \Ga_{(*)}
     \big(\phi_*,\phi, Q_n^*\fQ_n\de\psi_*,Q_n^*\fQ_n\de\psi\big)
     -\Ga_{(*)}^{(1)}
     \big(\phi_*,\phi, Q_n^*\fQ_n\de\psi_*,Q_n^*\fQ_n\de\psi\big)\big\}
\end{align*}
Hence the bounds on  $\de\varphi_{(*)}^{(\ge 2)}$ and
$\de\varphi_{(*)}^{(+)}= \de\varphi_{(*)}^{(1)\squig}
  + \de\varphi_{(*)}^{(\ge 2)}$
follows from \cite[Proposition \propSUBeqnsoln.b and Remark \remSUBcB.a]{SUB} 
with $d_{\rm max}=3$ and
\begin{align*}
&\max_{1\le j\le r}\sfrac{1}{\la_j}\tn B_j\tn_{w_{\ka,\la}}
\le \tilde \GGa'_4 \|V\|_\fm(\wf_\phi+\wf_{\de\psi})\wf_{\de\psi}
\\
&\fc=\max_{1\le j\le r}\sfrac{1}{\la_j}\tn L_j\tn_{w_{\ka,\la}}+
3\max_{1\le j\le r}\sfrac{1}{\la_j}\tn B_j\tn_{w_{\ka,\la}}
\le \tilde \GGa'_4\{\|V\|_\fm(\wf_\phi+\wf_{\de\psi})^2+|\mu|\}
\end{align*}

\Item (b) We follow the same strategy as in Proposition \ref{propBGEphivepssoln}.b.
That is, we apply $\partial_\nu$ to \eqref{eqnBGEdevarphi} and use
the ``discrete product rule'' \eqref{eqnBGEprodruleTrip} and
\begin{equation}\label{eqnBGEcommutederiv}
\partial_\nu {S_n^*}
= S_{n,\nu}^{(+)}\partial_\nu\qquad
\partial_\nu  S_n
= S_{n,\nu}^{(-)}\partial_\nu\qquad
\partial_\nu Q^*_n \fQ_n
= Q_{n,\nu}^{(+)} \fQ_n\partial_\nu
\end{equation}
where $Q_{n,\nu}^{(+)}$ was defined in \cite[(\eqnPBSqnplusminus)]{POA} and
$S_{n,\nu}^{(\pm)}$ was defined in \cite[(\eqnPOGSnnudef) and (\eqnPOGSnnuzerodef)]{POA}.
(See \cite[Remark \remPBSunderivAlg\ and (\eqnPOGpartialS)]{POA}.)
Denoting  $\de\phi_{(*)}=\de\varphi_{(*)}(\phi_*,\phi,\de\psi_*,\de\psi)$,
this gives
\begin{equation}\label{eqnBGEpartialDephi}
\begin{aligned}
&\partial_\nu\de\phi_*
  &+S_{n,\nu}^{(+)}L_{11}(\partial_\nu\de\phi_*)
  &+S_{n,\nu}^{(+)}L_{12}(\partial_\nu\de\phi)
&=S_{n,\nu}^{(+)}f_*\\
&\partial_\nu\de\phi
  &+S_{n,\nu}^{(-)}L_{21}(\partial_\nu\de\phi_*)
  &+S_{n,\nu}^{(-)}L_{22}(\partial_\nu\de\phi)
&=S_{n,\nu}^{(-)}f
\end{aligned}
\end{equation}
where
\begin{align*}
L_{11}(\partial_\nu\de\phi_*)&= -\mu\partial_\nu\de\phi_*
             +2\cV'_*\big(\phi_*,\phi,\partial_\nu\de\phi_*\big)
            +\cV'_*\big(\partial_\nu\de\phi_*,T_\nu^{-1}\phi,
                     \de\phi_*+T_\nu^{-1}\de\phi_*\big)
     \\&\hskip0.5in 
      +2\cV_*'\big(\partial_\nu\de\phi_*,T_\nu^{-1}\de\phi,T_\nu^{-1}\phi_*\big)
       +\cV_*'\big(\partial_\nu\de\phi_*,T_\nu^{-1}\de\phi,
                            \de\phi_*+T_\nu^{-1}\de\phi_*\big)
\\
L_{12}(\partial_\nu\de\phi)&= \cV'_*\big(\phi_*,\partial_\nu\de\phi,\phi_*\big)
      +2\cV'_*\big(\de\phi_*,\partial_\nu\de\phi,T_\nu^{-1}\phi_*\big)
      +\cV'_*\big(\de\phi_*,\partial_\nu\de\phi,\de\phi_*\big)
\\
L_{21}(\partial_\nu\de\phi_*)&=\cV'\big(\phi,\partial_\nu\de\phi_*,\phi\big)
          +2\cV'\big(\de\phi,\partial_\nu\de\phi_*,T_\nu^{-1}\phi\big)
          +\cV'\big(\de\phi,\partial_\nu\de\phi_*,\de\phi\big)
\\
L_{22}(\partial_\nu\de\phi)&=-\mu\partial_\nu\de\phi
           +2\cV'\big(\phi,\phi_*,\partial_\nu\de\phi\big)
  +\cV'\big(\partial_\nu\de\phi,T_\nu^{-1}\phi_*,\de\phi+T_\nu^{-1}\de\phi\big)
     \\&\hskip0.5in 
    +2\cV'\big(\partial_\nu\de\phi,T_\nu^{-1}\de\phi_*,T_\nu^{-1}\phi\big) 
+\cV'\big(\partial_\nu\de\phi,T_\nu^{-1}\de\phi_*,\de\phi+T_\nu^{-1}\de\phi\big)
\\
f_*&= Q_{n,\nu}^{(+)} \fQ_n[T_\nu\de\psi_*-\de\psi_*]
   -\cV'_*\big(\partial_\nu\phi_*,T_\nu^{-1}\de\phi,\phi_*+T_\nu^{-1}\phi_*\big)
     \cr&\hskip0.5in 
   -2\cV'_*\big(\partial_\nu\phi_*,T_\nu^{-1}\phi,T_\nu^{-1}\de\phi_*\big)
   -2\cV'_*\big(\phi_*,\partial_\nu\phi,T_\nu^{-1}\de\phi_*\big)
     \\&\hskip0.5in 
   -\cV'_*\big(\de\phi_*,\partial_\nu\phi,\de\phi_*\big)
   -2\cV'_*\big(\de\phi_*,\de\phi,\partial_\nu\phi_*\big)
\\
f&= Q_{n,\nu}^{(+)} \fQ_n[T_\nu\de\psi-\de\psi] 
     -\cV'\big(\partial_\nu\phi,T_\nu^{-1}\de\phi_*,\phi+T_\nu^{-1}\phi\big)
     \\&\hskip0.5in 
     -2\cV'\big(\partial_\nu\phi,T_\nu^{-1}\phi_*,T_\nu^{-1}\de\phi\big)
     -2\cV'\big(\phi,\partial_\nu\phi_*,T_\nu^{-1}\de\phi\big)   
     \\&\hskip0.5in 
     -\cV'\big(\de\phi,\partial_\nu\phi_*,\de\phi\big)
     -2\cV'\big(\de\phi,\de\phi_*,\partial_\nu\phi\big)
\end{align*}
The system of equations \eqref{eqnBGEpartialDephi} is of the form
\begin{equation*}
\vec\ga
=\vec f(\vec\al)+\vec L(\vec\al,\vec\ga)+\vec B\big(\vec\al\,;\, \vec\ga\big)
\end{equation*}
as in \cite[(\eqnSUBfixedpteqn.b)]{SUB}, with $\vec\al=(\al_1,\cdots,\al_8)$,
$\vec\ga=(\ga_*,\ga)$ and
\begin{align*}
\al_1&=\phi_*   
     &  \al_2&= \phi  
     &  \al_3&= \partial_\nu\phi_* \qquad
     &  \al_4&= \partial_\nu\phi  \\
         \al_5&=\de\phi_*   
     &  \al_6&= \de\phi 
     &  \al_7&=\de\psi_*   
     &  \al_8&= \de\psi \cr
     \partial_\nu\de\phi_*&=S_{n,\nu}^{(+)} \ga_* &
     \partial_\nu\de\phi&=S_{n,\nu}^{(-)}\ga 
\end{align*}
and $\vec B(\vec\al;\vec\ga)=0$ and
\begin{align*}
\vec L(\vec\al;\vec\ga)
&=-\left[\begin{matrix}
          L_{11}\, (S_{n,\nu}^{(+)} \ga_*)
          \ +\ L_{12}\,(S_{n,\nu}^{(-)} \ga)  \\
           L_{21}\,(S_{n,\nu}^{(+)} \ga_*)
         \ +\ L_{22}\, (S_{n,\nu}^{(-)} \ga)
          \end{matrix}\right]
\end{align*}
Now apply  \cite[Proposition \propSUBeqnsoln.a]{SUB} with $\fc=\half$ and
\begin{align*}
\ka_1& = \ka_2 = \wf_\phi  &
\ka_3& = \ka_4 = \wf_\phi' &
\ka_5& = \ka_6 = \GGa'_4\wf_{\de\psi} &
\ka_7& =\ka_8 = \wf_{\de\psi} \\
\la_1&=\la_2=4\wf_f 
\end{align*}
with 
\begin{align*}
\wf_f&= (e^\fm+1)\| Q_{n,\nu}^{(+)} \fQ_n\|_{\fm}\ka_7
  +e^{2\veps_n\fm}\|V\|_\fm\big\{6\ka_1\ka_3\ka_5
                             +3\ka_3\ka_5^2\big\}\\
    &= \Big[(e^\fm+1)\| Q_{n,\nu}^{(+)} \fQ_n\|_{\fm}
        +\GGa'_4e^{2\veps_n\fm}\|V\|_\fm\big\{
          6\wf_\phi\wf'_\phi+3\wf_\phi'\,\GGa'_4\wf_{\de\psi}
                      \big\}\Big]\wf_{\de\psi}\\
    &\le\half \GGa'_4\wf_{\de\psi}
\end{align*}
for a new $\GGa'_4$ and $\veps_n=\sfrac{1}{L^n}$.
Since $\tn f_j\tn_w \le \wf_f=\sfrac{1}{4}\la_j$,
$\tn B_j\tn_{w_{\ka,\la}}=0$ and 
\begin{align*}
\tn L_j\tn_{w_{\ka,\la}}
  &\le \max_{\si=+,-}\|S_{n,\nu}^{(\si)}\|_{\fm}
     \Big[|\mu| \la_j
            +\|V\|_\fm e^{2\veps_n\fm}
                 \big\{2\ka_1^2+4\ka_1\ka_5+2\ka_5^2\big\}\la_j\\&\hskip2.2in
             +\|V\|_\fm e^{\veps_n\fm}\big\{\ka_1^2+2\ka_1\ka_5+\ka_5^2\big\}
             \la_{3-j}
                    \Big]
\\
  &\le \max_{\si=+,-}\|S_{n,\nu}^{(\si)}\|_{\fm}
     \Big[|\mu|
      +3e^{2\veps_n\fm}\|V\|_\fm(\wf_\phi+\GGa'_4\wf_{\de\psi})^2\Big]
  \la_j
\end{align*}
\cite[Propositions \propSUBeqnsoln.a]{SUB} gives
\begin{align*}
\partial_\nu\de\phi_*
  &= S_{n,\nu}^{(+)}  \Ga_1\big(\al_1,\cdots, \al_8\big)\\
\partial_\nu\de\phi
  &= S_{n,\nu}^{(-)}\Ga_2 \big(\al_1,\cdots, \al_8\big)
\end{align*}
with 
\begin{align*}
&\tn \Ga_1\tn_{w_{\ka,\la}},
\tn \Ga_2\tn_{w_{\ka,\la}}
\le 2\max_{j=1,2}\tn f_j\tn_w
\le 2\wf_f
\le \GGa'_4\wf_{\de\psi}
\end{align*}
The conclusion now follows by \cite[Corollary \corSUBsubstitution]{SUB}. 

\end{proof}

\noindent
We define, on the space of field maps
$F(\phi_*,\phi,\de\psi_*,\de\psi)$, the projections
\begin{itemize}[leftmargin=*, topsep=2pt, itemsep=2pt, parsep=0pt]
\item[$\circ$]
$P^\phi_2$ 
which extracts the part of degree exactly one in each of $\phi_*$ 
and $\phi$, and of arbitrary degree in $\de\psi_{(*)}$ and
\item[$\circ$]
$P^\phi_1$ 
which extracts the part of degree exactly one in $\phi_{(*)}$, 
and of arbitrary degree in $\de\psi_{(*)}$  and
\item[$\circ$] 
$P^\phi_0$  which extracts the part of degree zero in 
$\phi_{(*)}$ and of arbitrary degree in $\de\psi_{(*)}$. 
\end{itemize}
\pagebreak[2]

\begin{lemma}\label{lemBGEdephisolnB}
Under the hypothesis of Lemma \ref{lemBGEdephisoln},
there is a constant $\GGa''_4$ such that the field maps 
$\de\varphi_{(*)}(\phi_*,\phi,\de\psi_*,\de\psi)$ of Lemma \ref{lemBGEdephisoln}
have the form
\begin{align*}
 \de\varphi_*
&= S_n(\mu)^* Q_n^* \fQ_n\, \de\psi_*  
  - {S_n(\mu)^*}\cV'_*(\varphi_*, \varphi,\varphi_*) 
      \Big|^{\atop{\varphi_*=\phi_*+S_n(\mu)^* Q_n^* \fQ_n\, \de\psi_*}
                  {\varphi=\phi+S_n(\mu) Q_n^* \fQ_n\, \de\psi}}
     _{\atop{\varphi_*=\phi}
            {\varphi=\phi}} 
   + \de\varphi_*^{(\ge 5)}
 \\
\noalign{\vskip0.05in}
\de\varphi
&= S_n(\mu) Q_n^* \fQ_n\, \de\psi 
- S_n(\mu)\cV' (\varphi, \varphi_*,\varphi) 
      \Big|^{\atop{\varphi_*=\phi_*+S_n(\mu)^* Q_n^* \fQ_n\, \de\psi_*}
                  {\varphi=\phi+S_n(\mu) Q_n^* \fQ_n\, \de\psi}}
     _{\atop{\varphi_*=\phi_*}
            {\varphi=\phi}}
   + \de\varphi^{(\ge 5)}
\end{align*}
with $\de\varphi_{(*)}^{(\ge 5)}$ being of order at least five in
$(\phi_{(*)},\de\psi_{(*)})$ and obeying 
\begin{align*}
  \TN P^\phi_j\de\varphi^{(\ge 5)}_{(*)} \TN_\phi   
     &\le \GGa''_4 \|V\|_\fm^2(\wf_\phi+\wf_{\de\psi})^j\wf_{\de\psi}^{5-j}
 \qquad\text{for $j=0,1,2$} \cr
\end{align*}

\end{lemma}
\begin{proof}
Rewrite the equations \eqref{eqnBGEdevarphi}
for $\de\varphi_{(*)}(\phi_*,\phi,\de\psi_*,\de\psi)$ in the form
\begin{align*}
 \de\varphi_*
&={S_n(\mu)^*}Q_n^* \fQ_n\, \de\psi_*  
  - {S_n(\mu)^*}\cV'_*(\varphi_*, \varphi,\varphi_*) 
      \Big|^{\atop{\varphi_*=\phi_*+\de\varphi_*}
                  {\varphi=\phi+\de\varphi}}
     _{\atop{\varphi_*=\phi}
            {\varphi=\phi}}
 \cr
\noalign{\vskip0.05in}
\de\varphi
&= S_n(\mu) Q_n^* \fQ_n\, \de\psi 
- S_n(\mu)\cV' (\varphi, \varphi_*,\varphi) 
      \Big|^{\atop{\varphi_*=\phi_*+\de\varphi_*}
                  {\varphi=\phi+\de\varphi}}
     _{\atop{\varphi_*=\phi_*}
             {\varphi=\phi}}
\end{align*}
We see from these equations that 
$\de\varphi_{(*)}=S_n(\mu)^{(*)} Q_n^* \fQ_n\, \de\psi_{(*)}
  + \de\varphi_{(*)}^{(\ge 3)}$, with $\de\varphi_{(*)}^{(\ge 3)}$ 
being of order at least three in $(\phi_{(*)},\de\psi_{(*)})$ and obeying
$\TN \de\varphi_{(*)}^{(\ge 3)} \TN_\fm\le \tilde\GGa''_4\|V\|_\fm$
\begin{align*}
  \TN P^\phi_j\de\varphi^{(\ge 3)}_{(*)} \TN_\phi   
        &\le \tilde\GGa''_4 \|V\|_\fm
                (\wf_\phi+\wf_{\de\psi})^j\wf_{\de\psi}^{3-j}
  \qquad\text{for $j=0,1,2$}
\end{align*}
Hence
\begin{align*}
 \de\varphi_*
&=S_n(\mu)^* Q_n^* \fQ_n\, \de\psi_*  
  - {S_n(\mu)^*}\cV'_*(\varphi_*, \varphi,\varphi_*) 
      \Big|^{\atop{\varphi_*=\phi_*+S_n(\mu)^* Q_n^* \fQ_n\, \de\psi_*
                  + \de\varphi_*^{(\ge 3)}}
        {\varphi=\phi+S_n(\mu) Q_n^* \fQ_n\, \de\psi+ \de\varphi^{(\ge 3)}}}
     _{\atop{\varphi_*=\phi}
             {\varphi=\phi}}
\end{align*}
The claim follows immediately from this and the corresponding equation
for $\de\varphi$.
\end{proof}

\begin{proof}[Proof of Proposition \ref{propBGEdephisoln}]
\Item {\it Parts (a) and (e):}\ \ \ 
By \eqref{eqnBGEdevarphi}
\begin{align*}
&\de\check\phi_{(*)n+1}(\th_*,\th,\de\psi_*,\de\psi,\mu,\cV)\\
&\hskip0.5in
   =\de\varphi_{(*)}\big(\check\phi_{*n+1}(\th_*,\th,\mu,\cV)\,,\,
        \check\phi_{n+1}(\th_*,\th,\mu,\cV)\,,\,\de\psi_*\,,\,\de\psi\big)
\end{align*}
so that, by \eqref{eqnBGEdehatphidef} and 
\cite[Definition \defBGAphicheck]{PAR1},
\begin{equation}\label{eqnBGEdevarphiB}
\begin{split}
&\de\hat\phi_{(*)n+1}(\psi_*,\psi,z_*,z)
=\bbbs\Big[\de\varphi_{(*)}\big(
           \phi_*\,,\,
           \phi\,,\,
           \de\psi_*\,,\,
           \de\psi\big)\Big]_{\atop
  {\phi_{(*)}=\check\phi_{(*)n+1}(\bbbs^{-1}\psi_*,\bbbs^{-1}\psi,\mu,\cV)}
  {\de\psi_{(*)}=D^{(n)(*)}\bbbl_* z_{(*)}}}\\
&\hskip0.5in=L^{\sfrac{3}{2}}\bbbl_*^{-1}\Big[\de\varphi_{(*)}\big(
           \bbbs^{-1}\Phi_*\,,\,
           \bbbs^{-1}\Phi\,,\,
           \bbbs^{-1} \de\Psi_*\,,\,
           \bbbs^{-1} \de\Psi\big)\Big]_{\atop
              {\Phi_{(*)}=\phi_{(*)n+1}(\psi_*,\psi,L^2\mu,\bbbs\cV)}
              {\de\Psi_{(*)}=L^{3/2}\bbbs D^{(n)(*)}\bbbs^{-1} z_{(*)}}}\\
&\hskip0.5in=L^{\sfrac{3}{2}}\de\varphi_{(*)}^{(s)}\big(
                     \Phi_*\,,\,
                     \Phi\,,\,
                     \de\Psi_*\,,\,
                     \de\Psi\big)\Big|_{\atop
              {\Phi_{(*)}=\phi_{(*)n+1}(\psi_*,\psi,L^2\mu,\bbbs\cV)}
              {\de\Psi_{(*)}=L^{3/2}\bbbs D^{(n)(*)}\bbbs^{-1} z_{(*)}}}
\end{split}
\end{equation}
in the notation of \cite[(\eqnSAscaledfm)]{PAR1}.
Similarly, using \cite[Remark \remSCscaling.b]{PAR1},
\begin{align*}
&\big(\partial_\nu\de\hat\phi_{(*)n+1}\big)(\psi_*,\psi,z_*,z)
=\bbbs_\nu \partial_\nu\bbbs^{-1}\de\hat\phi_{(*)n+1}(\psi_*,\psi,z_*,z)\\
&\hskip0.2in=L_\nu  L^{\sfrac{3}{2}}
          \bbbl_*^{-1}\Big[\partial_\nu\de\varphi_{(*)}\big(
           \bbbs^{-1}\Phi_*\,,\,
           \bbbs^{-1}\Phi\,,\,
           \bbbs^{-1} \de\Psi_*\,,\,
           \bbbs^{-1} \de\Psi\big)\Big]_{\atop
               {\Phi_{(*)}=\phi_{(*)n+1}(\psi_*,\psi,L^2\mu,\bbbs\cV)}
               {\de\Psi_{(*)}=L^{3/2}\bbbs D^{(n)(*)}\bbbs^{-1} z_{(*)}}}\\
&\hskip0.2in=L_\nu  L^{\sfrac{3}{2}}
          \bbbl_*^{-1}\Big[\de\varphi_{(*)\nu}\big(
           \bbbs^{-1}\Phi_*\,,\,
           \bbbs^{-1}\Phi\,,\,
           \bbbs^{-1}_\nu\partial_\nu\Phi_*\,,\,
           \bbbs^{-1}_\nu\partial_\nu\Phi\,,\,
           \bbbs^{-1} \de\Psi_*\,,\,
           \bbbs^{-1} \de\Psi\big)\Big]_{
                  \atop{\Phi_{(*)}=\cdots}
                       {\de\Psi_{(*)}=\cdots}}\\
&\hskip0.2in=\de\hat\phi_{(*)n+1,\nu}\big(\psi_*,\psi,\partial_\nu\psi_*,
                \partial_\nu\psi,z_*,z\big)\\
\end{align*}
where $L_0=L^2$ and $L_\nu=L$ for $1\le\nu\le 3$, and we have set
\begin{align*}
&\de\hat\phi_{(*)n+1,\nu}\big(\psi_*,\psi,\psi_{*\nu},\psi_\nu,
                      z_*,z\big)\\
&=L_\nu L^{\sfrac{3}{2}}
           \de\varphi_{(*)\nu}^{(s)}\big(
                     \Phi_*,
                     \Phi,
                     \Phi_{*\nu},
                     \Phi_\nu,
                     \de\Psi_*,
                     \de\Psi\big)\Big|_{\atop
            {\atop{\Phi_{(*)}=\phi_{(*)n+1}(\psi_*,\psi,L^2\mu,\bbbs\cV)}
                  {\Phi_{(*)\nu}=B^{(\pm)}_{n+1,L^2\mu,\nu}\!\psi_{(*)\nu}
                    +\phi^{(\ge 3)}_{(*)n+1,\nu}(\psi_*,\!\psi, 
                       \psi_{*\nu},\psi_\nu,
                       L^2\mu,\bbbs\cV)}}
                    {\de\Psi_{(*)}=L^{3/2}\bbbs D^{(n)(*)}\bbbs^{-1} z_{(*)}}}
\end{align*}

We shall bound $\de\varphi_{(*)}^{(s)}$ and 
$\de\varphi_{(*)\nu}^{(s)}$
using the norm $\tn\ \cdot\ \tn_\Phi$
with mass $\fm$ and weight factors 
\begin{align*}
\wf_\Phi&= \|S_{n+1}(L^2\mu)\|_\fm \|Q_{n+1}^* \fQ_{n+1}\|_\fm\wf
                 +\GGa_1 \|\bbbs V\|_{\fm}\wf^3\\
\wf'_\Phi&= \max_{\si=+,-}\|B^{(\si)}_{n+1,L^2\mu,\nu}\|_\fm\wf'
                 + \GGa_1 \|\bbbs V\|_\fm\wf^2\wf'\\
\wf_{\de\Psi}&=  L^{3/2}\|\bbbs D^{(n)}\bbbs^{-1}\|_\fm\ \wf_\fl
\end{align*}
By \cite[Corollary \corSUBsubstitution]{SUB} and Proposition \ref{propBGEphivepssoln},
with $n$ replaced by $n+1$, $\mu$ replaced by $L^2\mu$ and $V=\bbbs\cV$,
\begin{equation*}
\tn \de\hat\phi_{(*)n+1}\tn \le L^{3/2}\TN\de\varphi^{(s)}_{(*)}\TN_\Phi
\qquad
\tn \de\hat\phi_{(*)n+1,\nu}\tn \le L_\nu L^{3/2}
   \TN\de\varphi_{(*)\nu}^{(s)}\TN_\Phi
\end{equation*}
The hypothesis 
$\|\bbbs V\|_\fm\wf^2+L^2|\mu|\le\rrho_1$ of 
Proposition \ref{propBGEphivepssoln} is satisfied if $\rrho_2$ is 
small enough, since $\|\bbbs V\|_\fm\le \sfrac{1}{L}\|V\|_\fm$,
by \cite[Lemma \lemSAscaletoscale.a]{PAR1}. 
By \cite[Lemma \lemSAscaletoscale.c]{PAR1} with 
$\wf=\wf_\Phi$, $\wf'=\wf'_\Phi$, $\wf_\fl=\wf_{\de\Psi}$,
and $\check\fm=\fm$, 
$\check\wf=\wf_\phi$, $\check\wf'=\wf'_\phi$, $\check\wf_\fl=\wf_{\de\psi}$, with the choice
\begin{align*}
\wf_\phi&=L^{-3/2}\wf_\Phi
        = L^{-3/2}\big[\|S_{n+1}(L^2\mu)\|_\fm \|Q_{n+1}^* \fQ_{n+1}\|_\fm
                 +\GGa_1 \|\bbbs V\|_{\fm}\wf^2\big]\wf\\
\wf'_\phi&= L^{-5/2}\wf'_\Phi
   = L^{-5/2}\big[\max_{\si=+,-}\|B^{(\si)}_{n+1,L^2\mu,\nu}\|_\fm
                 + \GGa_1 \|\bbbs V\|_\fm\wf^2\big]\wf'\\
\wf_{\de\psi}&=L^{-3/2}\wf_{\de\Psi}
    =L^9\big(\sfrac{1}{L^9}\|\bbbs D^{(n)}\bbbs^{-1}\|_\fm\big)\ \wf_\fl
\end{align*}
we have
$
\tn\de\varphi^{(s)}_{(*)}\tn_\Phi
\le \tn\de\varphi_{(*)}\tn_\phi
$
and
$
 \TN\de\varphi_{(*)\nu}^{(s)}\TN_\Phi
\le  \TN \de\varphi_{(*)\nu}\TN_\phi
$
so that
\begin{equation}\label{eqnBGEdephisc}
\tn \de\hat\phi_{(*)n+1}\tn \le L^{3/2} \tn\de\varphi_{(*)}\tn_\phi
\qquad
\tn \de\hat\phi_{(*)n+1,\nu}\tn \le L_\nu L^{3/2}
   \TN \de\varphi_{(*)\nu}\TN_\phi
\end{equation}
So, by Lemma \ref{lemBGEdephisoln}, 
\begin{align*}
\tn \de\hat\phi_{(*)n+1}\tn &\le L^{3/2}\GGa'_4\wf_{\de\psi}
\le L^{11}\GGa_4\wf_\fl\\
\tn \de\hat\phi_{(*)n+1,\nu}\tn &\le L_\nu L^{11}
   \GGa_4\wf_\fl
\end{align*}
The hypothesis 
$\max\big\{\ |\mu|\ ,\  
 \|V\|_\fm (\wf_\phi+\wf_{\de\psi})
          (\wf_\phi+\wf'_\phi+\wf_{\de\psi})\ \big\}
\le \rrho'_2$ of Lemma \ref{lemBGEdephisoln} is satisfied if 
$\rrho_2$ is small enough.

\Item {\it Parts (b) and (c):}\ \ \ 
As in \eqref{eqnBGEdephisc},
\begin{alignat*}{3}
\tn \de\hat\phi_{(*)n+1}^{(+)}\tn 
           &\le L^{3/2} \tn\de\varphi_{(*)}^{(+)}\tn_\phi&
            &\le L^{29}\GGa_4\,
                    \{\|V\|_\fm(\wf+\wf_\fl)^2+|\mu|\}
                  \wf_\fl \\
\tn \de\hat\phi_{(*)n+1}^{(\ge 2)}\tn 
           &\le L^{3/2} \tn\de\varphi_{(*)}^{(\ge 2)}\tn_\phi&
            &\le L^{29}\GGa_4\,
                    \|V\|_\fm(\wf+\wf_\fl)
                  \wf_\fl^2 
\end{alignat*}
by Lemma \ref{lemBGEdephisoln}.a.

\Item {\it Part (d):}\ \ \ 
By \eqref{eqnBGEdevarphiB} and Lemma \ref{lemBGEdephisolnB},
\begin{align*}
\de\hat\phi_{n+1}^{(+)}(\psi_*,\psi,z_*,z)
&=L^{3/2}\bbbs [S_n(\mu)-S_n]Q_n^*\fQ_n 
                                   D^{(n)}\bbbs^{-1}z \\
&\hskip0.5in
  - L^{\sfrac{3}{2}}\bbbl_*^{-1}S_n(\mu)\cV'(\varphi, \varphi_*,\varphi) 
      \Big|^{\varphi_{(*)}=\phi_{(*)}
                 +S_n(\mu)^{(*)} Q_n^* \fQ_n\, \de\psi_{(*)}}
     _{\varphi_{(*)}=\phi_{(*)}}\\
&\hskip0.5in
    + L^{3/2}\bbbl_*^{-1}\de\varphi^{(\ge 5)}(\phi_*,\phi,\de\psi_*,\de\psi) 
\end{align*}
with the substitutions
\refstepcounter{equation}\label{eqnBGEsubB}
\begin{align}
\phi_{(*)}&=\bbbs^{-1}\phi_{(*)n+1}(\psi_*,\psi,L^2\mu,\bbbs\cV)
   \tag{\ref{eqnBGEsubB}.a}\\
\de\psi_{(*)}&=L^{3/2} D^{(n)(*)}\bbbs^{-1} z_{(*)}
   \tag{\ref{eqnBGEsubB}.b}
\end{align}
In the substitution, we expand, by Proposition \ref{propBGEphivepssoln}.a,
\begin{equation}\label{eqnBGEsubC}
\phi_{(*)}
=\bbbs^{-1}S_{n+1}(L^2\mu)^{(*)}Q_{n+1}^* \fQ_{n+1}\,\psi_{(*)}
      +\bbbs^{-1}\phi_{(*)n+1}^{(\ge 3)}(\psi_*,\psi,L^2\mu,\bbbs\cV)
\end{equation}
to get the statement of the proposition with 
$\de\hat\phi^{(\ho)}$ being the sum of 
\begin{equation*}
- L^{\sfrac{3}{2}}\bbbl_*^{-1}S_n(\mu)\cV'(\varphi, \varphi_*,\varphi) 
      \Big|^{\varphi_{(*)}=\phi_{(*)}
                      +S_n(\mu)^{(*)} Q_n^* \fQ_n\, \de\psi_{(*)}}
     _{\varphi_{(*)}=\bbbs^{-1}S_{n+1}(L^2\mu)^{(*)}Q_{n+1}^* \fQ_{n+1}\,\psi_{(*)}
          +S_n(\mu)^{(*)} Q_n^* \fQ_n\, \de\psi_{(*)}}
\end{equation*}
and
\begin{equation*}
L^{3/2}\bbbl_*^{-1}\de\varphi^{(\ge 5)}(\phi_*,\phi,\de\psi_*,\de\psi)
\end{equation*}
with the substitutions (\ref{eqnBGEsubB}.b) and \eqref{eqnBGEsubC}. 
As in \eqref{eqnBGEdephisc}, the specified
properties of $\de\hat\phi^{(\ho)}$ follow from 
\cite[Lemma \lemSAscaletoscale.c]{PAR1},
the properties of $\phi_{(*)n+1}^{(\ge 3)}$ in Proposition  
\ref{propBGEphivepssoln}.a\ 
and the properties of $\de\varphi^{(\ge 5)}$ in Lemma \ref{lemBGEdephisolnB}.

\end{proof}

\newpage
\section{
Variations of the Background Field with Respect to the Chemical
Potential $\mu$ and the Interaction $V$}\label{secBGEvrnsmu}

\begin{proposition}\label{propBGEdephidemu} 
There are constants\footnote{Recall Convention \ref{convBGEconstants}.}
$\rrho_3>0$ and $\GGa_5$, 
such that, if 
\begin{equation*}
\max\big\{\,|\mu|\,,\,|\de\mu|\,,\,\|V\|_\fm\wf^2\,,\,\|\de V\|_\fm \wf^2\, 
         \big\}
\le \rrho_3
\end{equation*}
then there are field maps 
  $\De\phi_{(*)n}$, 
  $\De\phi_{(*)n,\nu}$ and
  $\De\phi_{(*)n,D}$
such that
\begin{align*}
\phi_{(*)n}(\psi_*,\psi,\mu+\de\mu,\cV+\de\cV)
     &=\phi_{(*)n}(\psi_*,\psi,\mu,\cV)
       +\De\phi_{(*)n}(\psi_*,\psi,\mu,\de\mu,\cV,\de\cV)\\
\partial_\nu\phi_{(*)n}(\psi_*,\psi,\mu+\de\mu,\cV+\de\cV)
     &=\partial_\nu\phi_{(*)n}(\psi_*,\psi,\mu,\cV) \\ &\hskip0.75in
         +\De\phi_{(*)n,\nu}(\psi_*,\psi,\partial_\nu\psi_*,\partial_\nu\psi,
                 \mu,\de\mu,\cV,\de\cV)\\
D_n^{(*)}\phi_{(*)n}(\psi_*,\psi,\mu+\de\mu,\cV+\de\cV)
     &=D_n^{(*)}\phi_{(*)n}(\psi_*,\psi,\mu,\cV)  \\ &\hskip0.75in
         +\De\phi_{(*)n,D}\big(\psi_*,\psi,\mu,\de\mu,\cV,\de\cV\big)
\end{align*}
The field maps fulfill the bounds
\begin{align*}
\tn \De\phi_{(*)n}\tn &\le  \GGa_5\ \big(|\de\mu|+\|\de V\|_\fm\wf^2\big)\wf\\
\tn \De\phi_{(*)n,\nu}\tn 
                  &\le \GGa_5\ \big(|\de\mu|+\|\de V\|_\fm\wf^2\big)\wf'\\
\tn \De\phi_{(*)n,D}\tn &\le \GGa_5 \big(|\de\mu|+\|\de V\|_\fm\wf^2\big)\, \wf 
\end{align*}
Furthermore  $\De\phi_{(*)n}$ and $\De\phi_{(*)n,D}$ are of degree at least 
one in $\psi_{(*)}$
and each of  $\De\phi_{*n,\nu}$ and $\De\phi_{n,\nu}$ are of degree 
precisely one in $\psi_{(*)\nu}$.
Indeed, 
\begin{align*}
\De\phi_{(*)n}&=\de\mu\,B_{(*)n,\mu}\, \psi_{(*)}
                    +\De\phi_{(*)n}^{(\ge 3)}\\
\De\phi_{(*)n,\nu}&=\de\mu\,B_{(*)n,\mu,\nu}\,\psi_{(*)\nu}
                    +\De\phi_{(*)n,\nu}^{(\ge 3)}\\
\De\phi_{(*)n,D}&=\de\mu\,B_{(*)n,\mu,D}\,\psi_{(*)}
                     +\De\phi_{(*)n,D}^{(\ge 3)}
\end{align*}
where 
\begin{align*}
B_{(*)n,\mu}&=S_n^{(*)}\,
         \big[\bbbone -(\mu+\de\mu)S_n^{(*)}\big]^{-1} 
         S_n(\mu)^{(*)}Q_n^* \fQ_n\\
B_{*n,\mu,\nu}&=S_{n,\nu}^{(+)}\,
         \big[\bbbone -(\mu+\de\mu)S_{n,\nu}^{(+)}\big]^{-1} 
         B_{n,\nu,\mu}^{(+)}\\
B_{n,\mu,\nu}&=S_{n,\nu}^{(-)}\,
         \big[\bbbone -(\mu+\de\mu)S_{n,\nu}^{(-)}\big]^{-1} 
         B_{n,\nu,\mu}^{(-)}\\
B_{(*)n,\mu,D}&=S_n(\mu)^{(*)}Q_n^* \fQ_n 
         - \big(Q_n^*\fQ_nQ_n-\mu-\de\mu\big)
          B_{(*)n,\mu}
\end{align*} 
$\De\phi_{(*)n}^{(\ge 3)}$, $\De\phi_{(*)n,D}^{(\ge 3)}$ are 
of degree at least three in $\psi_{(*)}$.
$\De\phi_{(*)n,\nu}^{(\ge 3)}$ and $\De\phi_{n,\nu}^{(\ge 3)}$ are of degree
precisely one in $\psi_{(*),\nu}$ and of degree at least two in $\psi_{(*)}$.
They obey the bounds
\begin{align*}
\TN \De\phi_{(*)n}^{(\ge 3)}\TN\ ,\ 
\TN \De\phi_{(*)n,D}^{(\ge 3)}\TN
\ \le\  \GGa_5\big(|\de\mu|\|V\|_\fm+\|\de V\|_\fm\big)\, \wf^3\\
\TN \De\phi_{(*)n,\nu}^{(\ge 3)}\TN
\ \le\  \GGa_5\big(|\de\mu|\|V\|_\fm+\|\de V\|_\fm\big)\, \wf^2\wf'\\
\end{align*}

\end{proposition}

As in the proof of Lemma \ref{lemBGEdephisoln}, we fix any $\wf_\phi$ 
and $\wf'$, $0\le\nu\le 3$, and denote by $\tn\ \cdot\ \tn_\phi$ 
the (auxiliary) norm with mass $\fm$ that assigns the weight factors
$\wf_\phi$ to the fields $\phi_{(*)}$ and
$\wf'$ to the fields $\phi_{\nu(*)}$.

\begin{lemma}\label{lemBGEdephi} 
There are constants $\rrho'_3>0$, $\GGa'_5$, 
such that, if 
$$
\max\big\{\,|\mu|\,,\,|\de\mu|\,,\,
             \|V\|_\fm\wf_\phi^2\,,\,\|\de V\|_\fm \wf_\phi^2\, \big\}
\le \rrho'_3
$$
then the following are true. 
\begin{enumerate}[label=(\alph*), leftmargin=*]
\item[$\circ$]
There are field maps $\De\varphi_{(*)n}
   =\De\varphi_{(*)n}\big(\phi_*,\phi,\mu,\de\mu,\cV,\de\cV\big)$ 
such that
\begin{align*}
\phi_{*n}(\psi_*,\psi,\mu\!+\!\de\mu,\cV\!+\!\de\cV)
     &=\phi_{*n}(\psi_*,\psi,\mu,\cV) \\ &\hskip0.75in
        +\De\varphi_{*n}\big(\phi_*,\phi,\mu,\de\mu,\!\cV,\de\cV\big)
          \Big|_{\phi_{(*)}=\phi_{(*)n}(\psi_*,\psi,\mu,\cV)}\\
\phi_n(\psi_*,\psi,\mu\!+\!\de\mu,\cV\!+\!\de\cV)
     &=\phi_n(\psi_*,\psi,\mu,\cV) \\ &\hskip0.75in
      +\De\varphi_n\big(\phi_*,\phi,\mu,\de\mu,\cV,\de\cV\big)
          \Big|_{\phi_{(*)}=\phi_{(*)n}(\psi_*,\psi,\mu,\cV)}
\end{align*}
and
\begin{equation*}
\tn \De\varphi_{(*)n}\tn_\phi
\le 4\|S_n\|_{\fm}\big(|\de\mu|+\|\de V\|_\fm\wf_\phi^2\big)\, \wf_\phi 
\end{equation*}
Furthermore $\De\varphi_{*n}$ and $\De\varphi_n$ are of degree at least one in $\phi_{(*)}$.
Indeed
\begin{equation}\label{eqnBGEfgethree}
\De\varphi_{(*)n}=\de\mu\,S_n^{(*)}\,
         \big[\bbbone -(\mu+\de\mu)S_n^{(*)}\big]^{-1} \phi_{(*)}
             +\De\varphi_{(*)n}^{(\ge 3)}
\end{equation}
where $\De\varphi_{(*)n}^{(\ge 3)}$ is the part of $\De\varphi_{(*)n}$ that 
is of degree at least three in $\phi_{(*)}$, and 
\begin{align*}
\TN \De\varphi_{(*)n}^{(\ge 3)}\TN_\phi
&\le 4\|S_n\|_\fm\big\{\|\de V\|_\fm 
   +16\|S_n\|_\fm\ \|V\|_\fm\,|\de\mu|\big\}\wf_\phi^3
\end{align*}

\item[$\circ$]
There are field maps $\De\varphi_{(*)n,\nu}
   =\De\varphi_{(*)n,\nu}\big(\phi_*,\phi,\phi_{*\nu},\phi_\nu,
                     \mu,\de\mu,\cV,\de\cV\big)$
such that
\begin{align*}
&\partial_\nu\phi_{(*)n}(\psi_*,\psi,\mu+\de\mu,\cV+\de\cV) \\&\hskip0.4in
     =\partial_\nu\phi_{(*)n}(\psi_*,\psi,\mu,\cV)  \\ &\hskip0.8in
     +\De\varphi_{(*)n,\nu}\big(\phi_*,\phi,\partial_\nu\phi_*,\partial_\nu\phi,
                    \mu,\de\mu,\cV,\de\cV\big)
          \Big|_{\phi_{(*)}=\phi_{(*)n}(\psi_*,\psi,\mu,\cV)}
\end{align*}
and
\begin{equation*}
\tn \De\varphi_{(*)n,\nu}\tn_\phi
\le \GGa'_5\ \big(|\de\mu|+\|\de V\|_\fm\wf_\phi^2\big)\wf'_\phi
\end{equation*}
Furthermore $\De\varphi_{*n,\nu}$ and $\De\varphi_{n,\nu}$ are both 
of degree precisely one in $\phi_{(*)\nu}$. Indeed
\begin{equation}\label{eqnBGEfgethreenu}
\De\varphi_{(*)n,\nu}=\de\mu\,S_{n,\nu}^{(\pm)}\,
         \big[\bbbone -(\mu+\de\mu)S_{n,\nu}^{(\pm)}\big]^{-1} \phi_{(*)\nu}
             +\De\varphi_{(*)n,\nu}^{(\ge 3)}
\end{equation}
where $\De\varphi_{*n,\nu}^{(\ge 3)}$ and
$\De\varphi_{n,\nu}^{(\ge 3)}$ are both of degree precisely one in $\phi_{(*)\nu}$ 
and of degree at least two in $\phi_{(*)}$, and 
\begin{align*}
\TN \De\varphi_{(*)n,\nu}^{(\ge 3)}\TN_\phi
&\le \GGa'_5\ \big(|\de\mu| \|V\|_\fm+\|\de V\|_\fm\big)\,
                                 \wf_\phi^2\wf'_\phi
\end{align*}
\end{enumerate}
\end{lemma}

\begin{proof} (a)
Write  
\begin{align*}
\phi_{(*)n}(\psi_*,\psi,\mu+\de\mu,\cV+\de\cV)
&=\phi_{(*)n}(\psi_*,\psi,\mu,\cV)
          +\De\phi_{(*)}(\psi_*,\psi,\mu,\de\mu,\cV,\de\cV)\\
&=\phi_{(*)}+\De\phi_{(*)}
\end{align*}
Then, by \eqref{eqnBGEbgeqns}, using the notation of 
\cite[Definition \defBGAgradV]{PAR1},
\begin{alignat*}{3}
&S^{* -1}_n\big(\phi_*+\De\phi_*\big) 
  +(\cV'_*+\de\cV'_*)(\phi_*+\De\phi\,,\,
            \phi+\De\phi\,,\,\phi_*+\De\phi_*)\\
&\hskip2.8in -(\mu+\de\mu)[\phi_*+\De\phi_*]
       &&=Q_n^* \fQ_n\psi_*\\
&S^{-1}_n\big(\phi+\De\phi\big)
         +(\cV'+\de\cV')(\phi+\De\phi\,,\,
            \phi_*+\De\phi_*\,,\,\phi\De\phi)\\
&\hskip2.8in -(\mu+\de\mu)[\phi+\De\phi]
              &&=Q_n^* \fQ_n\psi
\end{alignat*}
Subtracting these equations but with 
$\de\mu=\de\cV=\de\cV'_{(*)}=\De\phi_{(*)}=0$, we see that $\De\phi_{(*)}=\De\phi_{(*)}(\psi_*,\psi,\mu,\de\mu,\cV,\de\cV)$
is the solution to
\begin{equation}\label{eqnBGEDephi}
\begin{alignedat}{5}
&\tilde S^{* -1}_n\De\phi_*
       &&+(\cV'_*+\de\cV'_*)(\phi_*\!+\!\De\phi_*\,,\,
            \phi+\De\phi\,,\,\phi_*\!+\!\De\phi_*)
       &&-\cV'_*(\phi_*\,,\,\phi\,,\,\phi_*)
  &&=\de\mu\,\phi_*\\
&\tilde S_n^{-1}\De\phi
       &&+(\cV'+\de\cV')(\phi\!+\!\De\phi\,,\,
            \phi_*+\De\phi_*\,,\,\phi\!+\!\De\phi)
       &&-\cV'(\phi\,,\,\phi_*\,,\,\phi)
   &&=\de\mu\,\phi
\end{alignedat}
\end{equation}
when
\begin{equation*}
\tilde S_n^{-1}= S_n^{-1} -\mu-\de\mu\qquad
\phi_*=\phi_{*n}(\mu,\cV)\qquad
\phi=\phi_n(\mu,\cV)
\end{equation*}
(Recall that $\tilde S_n^*$ is the transpose, rather than the adjoint, of
$\tilde S_n$.)
If $\rrho'_3$ is small enough $|\mu+\de\mu|\,\|S_n\|_{\fm}\le\half$, 
and $\|\tilde S_n\|_{\fm}\le 2\|S_n\|_{\fm}$. Rewrite \eqref{eqnBGEDephi} as
\begin{align*}
&\tilde S^{* -1}_n\De\phi_*
       &+(\cV'_*+\de\cV'_*)(\phi_*\!+\!\De\phi_*\,,\,
            \phi+\De\phi\,,\,\phi_*\!+\!\De\phi_*)
       &-(\cV'_*+\de\cV'_*)(\phi_*\,,\,\phi\,,\,\phi_*) \cr
& &  &\hskip0.25in=\de\mu\,\phi_*-\de\cV'_*(\phi_*\,,\,\phi\,,\,\phi_*)\\
&\tilde S_n^{-1}\De\phi
       &+(\cV'+\de\cV')(\phi\!+\!\De\phi\,,\,
            \phi_*+\De\phi_*\,,\,\phi\!+\!\De\phi)
       &-(\cV'+\de\cV')(\phi\,,\,\phi_*\,,\,\phi) \\
& &   &\hskip0.25in=\de\mu\,\phi-\de\cV'(\phi\,,\,\phi_*\,,\,\phi)\\
\end{align*}
This is of the form
\begin{equation*}
\vec\ga
=\vec f(\vec\al)+\vec L(\vec\al,\vec\ga)+\vec B\big(\vec\al\,;\, \vec\ga\big)
\end{equation*}
as in \cite[(\eqnSUBfixedpteqn.b)]{SUB}, with $X=X_n$ and
\begin{align*}
\al_*&=\phi_*   
     &  \al&= \phi
     & \de\al_*&=  \de\mu\,\phi_*    
     &  \de\al&= \de\mu\,\phi  
     & \vec\al& = \big(\al_*,\al,\de\al_*,\de\al\big)\\
& &
& &
        \De\phi_*&=\tilde S_n^* \ga_* &
        \De\phi&=\tilde S_n\ga   
     & \vec\ga&=\big(\ga_*,\ga\big) 
\end{align*}
and 
\begin{align*}
\vec f(\vec\al)(u)
&=\left[\begin{matrix}
                \de\al_*(u) - \de\cV'_*\big(u;\al_*,\al,\al_*\big) \\
                \de\al(u) - \de\cV'\big(u;\al,\al_*,\al\big)
        \end{matrix}\right]
\cr
\noalign{\vskip0.1in}
\vec L(\vec\al;\vec\ga)(u)
&=\left[\begin{matrix}
                  - (\cV'_*+\de\cV'_*)\big(u;\al_*,\tilde S_n \ga,\al_*\big)
               - 2(\cV'_*+\de\cV'_*)\big(u;\al_*,\al,\tilde S_n^*\ga_*\big) \\
           \noalign{\vskip0.05in}
               - (\cV'+\de\cV')\big(u;\al,\tilde S_n^*\ga_*,\al\big)
               - 2(\cV'+\de\cV')\big(u;\al,\al_*,\tilde S_n \ga\big)
          \end{matrix}\right]
\\
\noalign{\vskip0.1in}
\vec B(\vec\al;\vec\ga)(u)
&=\left[\begin{matrix}
   - (\cV'_*+\de\cV'_*)\big(u;\tilde S_n^*\ga_*,\al,\tilde S_n^*\ga_*\big)
   - 2(\cV'_*+\de\cV'_*)\big(u;\tilde S_n^*\ga_*,\tilde S_n \ga,\al_*\big)
\hfill\null\\\hskip2.5in
   - (\cV'_*+\de\cV'_*)\big(u;\tilde S_n^*\ga_*,\tilde S_n \ga,
                                           \tilde S_n^*\ga_*\big)\\
          \noalign{\vskip0.05in}
    - (\cV'+\de\cV')\big(u;\tilde S_n \ga,\al_*,\tilde S_n \ga\big)
   - 2(\cV'+\de\cV')\big(u;\tilde S_n \ga,\tilde S_n^*\ga_*,\al\big)
\hfill\null\\ \hskip2.5in
   - (\cV'+\de\cV')\big(u;\tilde S_n \ga,\tilde S_n^*\ga_*,\tilde S_n\ga\big)
      \end{matrix}\right]
\end{align*}
Now apply  \cite[Proposition \propSUBeqnsoln.a and Remark \remSUBcB.a]{SUB} 
with $d_{\rm max}=3$, $\fc=\half$ and
\begin{equation*}
\ka_1= \ka_2 = \wf_\phi \qquad
\ka_3=\ka_4=|\de\mu|\,\wf_\phi \qquad
\la_1=\la_2=4\ka_f
\end{equation*}
with
\begin{equation*}
\ka_f = |\de\mu|\,\wf_\phi +\|\de V\|_\fm\wf_\phi^3
\end{equation*}
Since
\begin{align*}
\tn f_j\tn_w &\le  \ka_f=\sfrac{1}{4}\la_j\\
\tn L_j\tn_{w_{\ka,\la}}&\le \|\tilde S_n\|_{\fm}\|V+\de V\|_\fm
     \{\ka_j^2\la_{3-j}+2\ka_1\ka_2\la_j\} \cr
  &\le  6\|S_n\|_{\fm}\|V+\de V\|_\fm \wf_\phi^2\la_j\\
\tn B_j\tn_{w_{\ka,\la}}
   &\le\|\tilde S_n\|_{\fm}^2\ \|V+\de V\|_\fm\   
           \big\{\ka_{3-j}\la_j^2+2\ka_j\la_1\la_2\big\}
            +\|\tilde S_n\|_{\fm}^3\|V+\de V\|_\fm\la_j^2\la_{3-j} \\
   &\le 4\|S_n\|_{\fm}^2\ \|V+\de V\|_\fm\  
            \big\{3\ka_j + 2\|S_n\|_{\fm}\la_j\big\}\la_j^2 \cr
   &\le 4\|S_n\|_{\fm}^2\ \|V+\de V\|_\fm\  
      \big\{3 + 8\|S_n\|_{\fm}\big(|\de\mu|+\|\de V\|_\fm\wf_\phi^2\big)\big\}
          \\&\hskip3.5in
              \ 4\wf_\phi^2 \big(|\de\mu|+\|\de V\|_\fm\wf_\phi^2\big) \la_j \\
  &\le  50\|S_n\|_{\fm}^2\|V+\de V\|_\fm
     \ \big(|\de\mu|+\|\de V\|_\fm\wf_\phi^2\big)\ \wf_\phi^2\la_j
\end{align*}
\cite[Proposition \propSUBeqnsoln.a]{SUB} gives
\begin{align*}
\De\phi_{*n}
  &= \tilde S_n^*  \Ga_1\big(\phi_*,\phi,\de\mu\,\phi_*,\de\mu\,\phi\big)\\
\De\phi_n
   &=\tilde S_n\Ga_2 \big(\phi_*,\phi,\de\mu\,\phi_*,\de\mu\,\phi\big)
\end{align*}
with 
\begin{align*}
\tn \Ga_1\tn_w,
\tn \Ga_2\tn_w
\le 2\ka_f
= 2\ \big(|\de\mu|+\|\de V\|_\fm\wf_\phi^2\big)\ \wf_\phi
\end{align*}

Now we prove \eqref{eqnBGEfgethree}, using the same system of equations and 
the same $\vec\al$, $\vec\ga$, $\ka$'s and $\la$'s. But we apply
\cite[Proposition \propSUBeqnsoln.b]{SUB} with
\begin{equation*}
\fc=\max_{j=1,2}\sfrac{1}{\la_j}\tn L_j\tn_{w_{\ka,\la}}
     +3\max_{j=1,2}\sfrac{1}{\la_j}\tn B_j\tn_{w_{\ka,\la}}
\le 8\|S_n\|_{\fm}\|V+\de V\|_\fm \wf_\phi^2
\end{equation*}
which gives, for $j=1,2$,
\begin{align*}
\tn \Ga^{(1)}_j- f_j\tn_w
       &\le\sfrac{\fc}{1-\fc}\max_{j'=1,2}\tn f_{j'}\tn_w
       \le 16\|S_n\|_{\fm}\|V+\de V\|_\fm 
             \big(|\de\mu|+\|\de V\|_\fm\wf_\phi^2\big)\, \wf_\phi^3\\
\tn \Ga_j-\Ga^{(1)}_j\tn_w
   &\le\max_{j'=1,2}\tn B_j\tn_{w_{\ka,\la}}
   \le 200\|S_n\|_{\fm}^2\|V+\de V\|_\fm
              \big(|\de\mu|+\|\de V\|_\fm\wf_\phi^2\big)^2 \wf_\phi^3
\end{align*}
where $\vec\Ga^{(1)}$ is the solution of
$\vec\ga=\vec f(\vec\al)+\vec L(\vec\al,\vec\ga)$.
Since 
\begin{equation*}
\vec \Ga^{(1)}
=\vec f(\vec\al)+\vec L(\vec\al,\vec\Ga^{(1)})
\qquad\text{and}\qquad
\vec \Ga
=\vec f(\vec\al)+\vec L(\vec\al,\vec\Ga)+\vec B\big(\vec\al\,;\, \vec\Ga\big)
\end{equation*}
and $\vec\Ga$ is degree at least one in $\phi_{(*)}$ and $\vec L$ and 
$\vec B$ are of degree three in $(\vec\al,\vec\ga)$, both  
$\vec\Ga^{(1)}-\vec f$ and $\vec\Ga-\vec f$ are of degree at least 3 in
$\phi_{(*)}$. So is $\vec f-\vec F$ where
\begin{equation*}
\vec F = \left[\begin{matrix}
                \de\al_* \\
                \de\al
                \end{matrix}\right]
       = \left[\begin{matrix}
                \de\mu\,\phi_*  \\
                \de\mu\,\phi
                \end{matrix}\right]
\end{equation*}
Consequently
\begin{align*}
\De\varphi_{(*)n}&=\de\mu\,\tilde S_n^{(*)}\, \phi_{(*)}
             +\De\varphi_{(*)}^{(\ge 3)}
\end{align*}
with
\begin{align*}
\De\varphi_{*n}^{(\ge 3)}
&=\tilde S_n^*\big\{[f_1-F_1]+[\Ga_1^{(1)}-f_1]+[\Ga_1-\Ga^{(1)}_1]\big\}\\
\De\varphi_n^{(\ge 3)}
&=\tilde S_n\big\{[f_2-F_2]+[\Ga_2^{(1)}-f_2]+[\Ga_2-\Ga^{(1)}_2]\big\}
\end{align*}
As $\tilde S_n^{(*)}=\big[{S_n^{(*)}}^{-1}-\mu-\de\mu\big]^{-1}
=S_n^{(*)}\,\big[\bbbone -(\mu+\de\mu)S_n^{(*)}\big]^{-1}$ and
\begin{align*}
&\tn f_j-F_j\tn_w
+\tn\Ga_j^{(1)}-f_j\tn_w
+\tn\Ga_j-\Ga^{(1)}_j\tn_w \\
&\hskip1in \le \|\de V\|_\fm \, \wf_\phi^3
    +32\|S_n\|_{\fm}\|V+\de V\|_\fm 
             \big(|\de\mu|+\|\de V\|_\fm\wf_\phi^2\big)\, \wf_\phi^3 \\
&\hskip1in \le 2\big\{\|\de V\|_\fm 
    +16\|S_n\|_{\fm}\|V\|_\fm |\de\mu|\big\} \wf_\phi^3 
\end{align*}
the desired bound on $\TN \De\varphi_{(*)n}^{(\ge 3)}\TN_\phi$
follows by \cite[Proposition \propSUBsubstitution.a]{SUB}.

\Item (b) 
By \eqref{eqnBGEcommutederiv},
applying $\partial_\nu$ to \eqref{eqnBGEDephi} gives
\begin{equation}\label{eqnMSpartialDePhi}
\begin{alignedat}{5}
&\big(\tilde S_{n,\nu}^{(+)}\big)^{-1}\partial_\nu\De\phi_*
  &+L_{11}(\partial_\nu\De\phi_*)
  &+L_{12}(\partial_\nu\De\phi)
&&=f_*\\
&\big(\tilde S_{n,\nu}^{(-)}\big)^{-1}\partial_\nu\De\phi
  &+L_{21}(\partial_\nu\De\phi_*)
  &+L_{22}(\partial_\nu\De\phi)
&&=f
\end{alignedat}
\end{equation}
where $\big(\tilde S_{n,\nu}^{(\si)}\big)^{-1}
=\big(S_{n,\nu}^{(\si)}\big)^{-1}-\mu-\de\mu$ and
\begin{align*}
L_{11}(\partial_\nu\De\phi_*)&= 
            2(\cV'_*+\de\cV'_*)\big(\phi_*,\phi,\partial_\nu\De\phi_*\big)
            +(\cV'_*+\de\cV'_*)\big(\partial_\nu\De\phi_*,T_\nu^{-1}\phi,
                     \De\phi_*\!+\!T_\nu^{-1}\De\phi_*\big)
     \\&\hskip0.4in  
           +2(\cV'_*+\de\cV'_*)\big(\partial_\nu\De\phi_*,
                                      T_\nu^{-1}\De\phi,T_\nu^{-1}\phi_*\big)
     \\&\hskip0.4in  
       +(\cV'_*+\de\cV'_*)\big(\partial_\nu\De\phi_*,T_\nu^{-1}\De\phi,
                            \De\phi_*+T_\nu^{-1}\De\phi_*\big)
\displaybreak[0]\\
L_{12}(\partial_\nu\De\phi)&= 
     (\cV'_*+\de\cV'_*)\big(\phi_*,\partial_\nu\De\phi,\phi_*\big)
      +2(\cV'_*+\de\cV'_*)\big(\De\phi_*,
                                \partial_\nu\De\phi,T_\nu^{-1}\phi_*\big)
     \cr&\hskip0.4in  
      +(\cV'_*+\de\cV'_*)\big(\De\phi_*,\partial_\nu\De\phi,\De\phi_*\big)
\displaybreak[0]\\
L_{21}(\partial_\nu\De\phi_*)&=
          (\cV'+\de\cV')\big(\phi,\partial_\nu\De\phi_*,\phi\big)
          +2(\cV'+\de\cV')\big(\De\phi,\partial_\nu\De\phi_*,T_\nu^{-1}\phi\big)
     \\&\hskip0.4in  
          +(\cV'+\de\cV')\big(\De\phi,\partial_\nu\De\phi_*,\De\phi\big)
\displaybreak[0]\\
L_{22}(\partial_\nu\De\phi)&=
        2(\cV'+\de\cV')\big(\phi,\phi_*,\partial_\nu\De\phi\big)
  +(\cV'+\de\cV')\big(\partial_\nu\De\phi,T_\nu^{-1}\phi_*,
                             \De\phi+T_\nu^{-1}\De\phi\big)
     \\&\hskip0.4in 
    +2(\cV'+\de\cV')\big(\partial_\nu\De\phi,T_\nu^{-1}\De\phi_*,
                                        T_\nu^{-1}\phi\big) 
     \\&\hskip0.4in  
+(\cV'+\de\cV')\big(\partial_\nu\De\phi,T_\nu^{-1}\De\phi_*,
                             \De\phi+T_\nu^{-1}\De\phi\big)
\displaybreak[0]\\
f_*&= \de\mu\,\partial_\nu\phi_*
   -(\cV'_*+\de\cV'_*)\big(\partial_\nu\phi_*,T_\nu^{-1}\De\phi,
                                            \phi_*+T_\nu^{-1}\phi_*\big)
     \\&\hskip0.4in 
   -2(\cV'_*+\de\cV'_*)\big(\partial_\nu\phi_*,T_\nu^{-1}\phi,
                                            T_\nu^{-1}\De\phi_*\big)
     \\&\hskip0.4in  
   -2(\cV'_*+\de\cV'_*)\big(\phi_*,\partial_\nu\phi,T_\nu^{-1}\De\phi_*\big)
     \\&\hskip0.4in 
   -(\cV'_*+\de\cV'_*)\big(\De\phi_*,\partial_\nu\phi,\De\phi_*\big)
   -2(\cV'_*+\de\cV'_*)\big(\De\phi_*,\De\phi,\partial_\nu\phi_*\big)
     \\&\hskip0.4in 
   -\de\cV'_*\big(\phi_*,\partial_\nu\phi,\phi_*\big)
   -\de\cV'_*\big(\partial_\nu\phi_*,T_\nu^{-1}\phi,\phi_*+T_\nu^{-1}\phi_*\big)
\displaybreak[0]\\
f&= \de\mu\,\partial_\nu\phi 
     -(\cV'+\de\cV')\big(\partial_\nu\phi,T_\nu^{-1}\De\phi_*,
                                            \phi+T_\nu^{-1}\phi\big)
     \\&\hskip0.4in 
     -2(\cV'+\de\cV')\big(\partial_\nu\phi,T_\nu^{-1}\phi_*,
                                             T_\nu^{-1}\De\phi\big)
     -2(\cV'+\de\cV')\big(\phi,\partial_\nu\phi_*,T_\nu^{-1}\De\phi\big)   
     \\&\hskip0.4in 
     -(\cV'+\de\cV')\big(\De\phi,\partial_\nu\phi_*,\De\phi\big)
     -2(\cV'+\de\cV')\big(\De\phi,\De\phi_*,\partial_\nu\phi\big)
     \\&\hskip0.4in 
   -\de\cV'\big(\partial_\nu\phi,T_\nu^{-1}\phi_*,\phi+T_\nu^{-1}\phi\big)
\end{align*}
Here we have used the ``discrete product rule'' \eqref{eqnBGEprodruleTrip}.
Observe that, if $\rrho'_3$ is small enough, then $|\mu+\de\mu|\,\|S_{n,\nu}^{(\si)}\|_{\fm}\le\half$, 
and $\|\tilde S_{n,\nu}^{(\si)}\|_{\fm}\le 2\|S_{n,\nu}^{(\si)}\|_{\fm}$.

The system of equations \eqref{eqnMSpartialDePhi} is of the form
\begin{equation*}
\vec\ga
=\vec f(\vec\al)+\vec L(\vec\al,\vec\ga)+\vec B\big(\vec\al\,;\, \vec\ga\big)
\end{equation*}
as in \cite[(\eqnSUBfixedpteqn.b)]{SUB}, with $\vec\al=(\al_1,\cdots,\al_6)$,
$\vec\ga=(\ga_*,\ga)$ and
\begin{align*}
        \al_1&=\phi_*   
     &  \al_2&= \phi  
     &  \al_3&= \partial_\nu\phi_* \quad
     &  \al_4&= \partial_\nu\phi   \quad
     &  \al_5&=\De\phi_*           \ \ 
     &  \al_6&= \De\phi \\
     \partial_\nu\De\phi_*&=\tilde S_{n,\nu}^{(+)} \ga_* &
     \partial_\nu\De\phi&=\tilde S_{n,\nu}^{(-)}\ga 
\end{align*}
and $\vec B(\vec\al;\vec\ga)=0$ and
\begin{align*}
\vec L(\vec\al;\vec\ga)
&=-\left[\begin{matrix}
          L_{11}\, (\tilde S_{n,\nu}^{(+)} \ga_*)
          \ +\ L_{12}\,(\tilde S_{n,\nu}^{(-)} \ga)  \\
           L_{21}\,(\tilde S_{n,\nu}^{(+)} \ga_*)
         \ +\ L_{22}\, (\tilde S_{n,\nu}^{(-)} \ga)
         \end{matrix}\right]
\end{align*} 
Now apply  \cite[Proposition \propSUBeqnsoln.a]{SUB} with $\fc=\half$ and
\begin{align*}
\ka_1&= \ka_2 = \wf_\phi  &
\ka_3&= \ka_4 = \wf'_\phi &
\ka_5&= \ka_6 =  4\|S_n\|_{\fm}\big(|\de\mu|+\|\de V\|_\fm\wf_\phi^2\big)\, 
                             \wf_\phi\\
\la_1&=\la_2=4\ka_f 
\end{align*}
with
\begin{align*}
\ka_f&=\ka_3\big\{|\de\mu|
      +\|V+\de V\|_\fm e^{2\veps_n\fm}
                 \big(6\ka_1+3\ka_5\big)\ka_5
       +\|\de V\|_\fm e^{2\veps_n\fm}
                 3\ka_1^2\big\} \\
     &\le 8 e^{2\veps_n\fm} \big\{|\de\mu|
       +\|\de V\|_\fm \wf_\phi^2\big\} \wf'_\phi
\end{align*}
and $\veps_n=\sfrac{1}{L^n}$.
Since $\|f_j\|_w \le \ka_f=\sfrac{1}{4}\la_j$,
$\|B_j\|_{w_{\ka,\la}}=0$ and 
\begin{align*}
\|L_j\|_{w_{\ka,\la}}
&\le \max_{\si=+,-}\|\tilde S_{n,\nu}^{(\si)}\|_{\fm}
     \|V+\de V\|_\fm e^{2\veps_n\fm}\Big[\big\{
     2\ka_1^2\!+\!4\ka_1\ka_5\!+\!2\ka_5^2
     \big\}\la_j
  \\ &\hskip3in
     +\big\{\ka_1^2\!+\!2\ka_1\ka_5\!+\!\ka_5^2\big\}\la_{3-j}\Big] \\
&\le 3\max_{\si=+,-}\|\tilde S_{n,\nu}^{(\si)}\|_{\fm}
     \|V+\de V\|_\fm e^{2\veps_n\fm}\big\{\ka_1+\ka_5\big\}^2\la_j\\
  &\le 3\max_{\si=+,-}\|\tilde S_{n,\nu}^{(\si)}\|_{\fm}
     \|V+\de V\|_\fm e^{2\veps_n\fm}
       \big\{1+4\|S_n\|_{\fm}\big(|\de\mu|+\|\de V\|_\fm\wf_\phi^2\big)\big\}^2
                 \wf_\phi^2 \la_j
\end{align*}
\cite[Propositions \propSUBeqnsoln.a]{SUB} gives
\begin{align*}
\partial_\nu\De\phi_*
  &= \tilde S_{n,\nu}^{(+)}  \Ga_1\big(\al_1,\cdots, \al_6\big)\\
\partial_\nu\De\phi
  &= \tilde S_{n,\nu}^{(-)}\Ga_2 \big(\al_1,\cdots, \al_6\big)
\end{align*}
with 
\begin{align*}
&\tn \Ga_1\tn_{w_{\ka,\la}},
\tn \Ga_2\tn_{w_{\ka,\la}}
\le 16 e^{2\veps_n\fm} \big\{|\de\mu|
       +\|\de V\|_\fm \wf_\phi^2\big\} \wf'_\phi
\end{align*}
The conclusions, except for \eqref{eqnBGEfgethreenu} now follow by 
\cite[Corollary \corSUBsubstitution]{SUB}. 

To prove \eqref{eqnBGEfgethreenu}, write 
\begin{equation*}
\left[\begin{matrix}
            \partial_\nu\De\varphi_{*n}\big(\phi_*,\phi,\mu,\de\mu\big)\\  
            \partial_\nu\De\varphi_n\big(\phi_*,\phi,\mu,\de\mu\big)
       \end{matrix}\right]
=\left[\begin{matrix}
                \tilde S^{(+)}_{n,\nu}&0\\
                0&\tilde S^{(-)}_{n,\nu}
         \end{matrix}\right]
\Big\{\vec f(\vec\al)+\vec L\Big(\vec\al\,,\,\big[\Ga_1(\vec\al),\Ga_2(\vec\al)\big]\Big)\Big\}
\end{equation*}
with
\begin{align*}
\al_1&=\phi_*  \qquad\qquad 
\al_2= \phi  & \qquad
\al_3&= \partial_\nu\phi_* \qquad\quad
\al_4= \partial_\nu\phi   \\
\al_5&=\De\varphi_{*n}\big(\phi_*,\phi,\mu,\de\mu,\cV,\de\cV\big)    &
\al_6&=\De\varphi_n\big(\phi_*,\phi,\mu,\de\mu,\cV,\de\cV\big) 
\end{align*}
Observe that the right hand side is of the form
\begin{align*}
\de\mu\left[\begin{matrix}
                    \tilde S^{(+)}_{n,\nu}\partial_\nu\phi_*\\
                    \tilde S^{(-)}_{n,\nu}\partial_\nu\phi\cr
             \end{matrix}\right]
+\left[\begin{matrix}
\De\varphi_{*n,\nu}^{(\ge 3)}
      \big(\phi_*,\phi,\partial_\nu\phi_*,\partial_\nu\phi,
                    \mu,\de\mu,\cV,\de\cV\big)\\
\De\varphi_{n,\nu}^{(\ge 3)}
      \big(\phi_*,\phi,\partial_\nu\phi_*,\partial_\nu\phi,
                    \mu,\de\mu,\cV,\de\cV\big)
\end{matrix}\right]
\end{align*}
with $\De\varphi_{(*)n,\nu}^{(\ge 3)}
       \big(\phi_*,\phi,\partial_\nu\phi_*,\partial_\nu\phi,
                       \mu,\de\mu,\cV\big)$
a finite sum of terms each of which is either of the form
$\pm\tilde S^{(\pm)}_{n,\nu}(\cV'_{(*)}+\de\cV'_{(*)})(\ze_1,\ze_2,\ze_3)$ 
with
\begin{itemize}[leftmargin=*, topsep=2pt, itemsep=2pt, parsep=0pt]
\item[$\circ$]
 exactly one of $\ze_1$, $\ze_2$, $\ze_3$ being one of 
$\partial_\nu\phi_{(*)}$, $\partial_\nu\De\phi_{(*)}$ (which are
of degree precisely one in $\partial_\nu\phi_{(*)}$) and
\item[$\circ$]
  each of the remaining two $\ze_j$'s being one of 
$\phi_{(*)}$, $\De\phi_{(*)}$, possibly translated
by $T_\nu^{-1}$, (which are of degree at least  one in $\phi_{(*)}$)  and
\item[$\circ$]
at least one of $\ze_1$, $\ze_2$, $\ze_3$ being one of 
$\De\phi_{(*)}$, $\partial_\nu\De\phi_{(*)}$, possibly translated
by $T_\nu^{-1}$.
\end{itemize}
or of the form
$\pm\tilde S^{(\pm)}_{n,\nu}\de\cV'_{(*)}(\ze_1,\ze_2,\ze_3)$ 
with
\begin{itemize}[leftmargin=*, topsep=2pt, itemsep=2pt, parsep=0pt]
\item[$\circ$]
exactly one of $\ze_1$, $\ze_2$, $\ze_3$ being a
$\partial_\nu\phi_{(*)}$   and
\item[$\circ$]
the remaining two $\ze_j$'s being a $\phi_{(*)}$, 
possibly translated by $T_\nu^{-1}$.
\end{itemize}
The degree properties and bounds on $\De\varphi_{(*)n,\nu}^{(\ge 3)}$
follow, with, in the bound,
\begin{itemize}[leftmargin=*, topsep=2pt, itemsep=2pt, parsep=0pt]
\item[$\circ$]
a factor of $\|V+\de V\|_\fm$ coming from the kernel of $\cV'_{(*)}+\de\cV'_{(*)}$,
\item[$\circ$]
a factor of $\|\de V\|_\fm$ coming from the kernel of $\de\cV'_{(*)}$,
\item[$\circ$]
$\partial_\nu\phi_{(*)}$ contributing a factor $\wf'_\phi$,
\item[$\circ$]
$\partial_\nu\De\phi_{(*)}$ contributing a factor
of $\const \big(|\de\mu|+\|\de V\|_\fm\wf_\phi^2\big) \wf'_\phi\le
\const  \wf'_\phi$,
\item[$\circ$]
 each $\phi_{(*)}$, possibly translated
by $T_\nu^{-1}$, contributing a factor of $\const\wf_\phi$, and
\item[$\circ$] 
each $\De\phi_{(*)}$, possibly translated by $T_\nu^{-1}$,
giving a factor of 
\begin{equation*}
\const \big(|\de\mu|+\|\de V\|_\fm\wf_\phi^2\big) \wf_\phi
\le\const\wf_\phi
\end{equation*}
\end{itemize}
since $\|V\|_\fm\|\de V\|_\fm\wf_\phi^2\le \const \|\de V\|_\fm$.
\end{proof}

\begin{proof}[Proof of Proposition \ref{propBGEdephidemu}]
We apply Lemma \ref{lemBGEdephi} with
\begin{align*}
\wf_\phi&=2\|S_n\|_\fm\|Q_n^* \fQ_n\|_\fm\wf
          +\GGa_1 \|V\|_{\fm}\wf^3
\\
\wf'_\phi&=\max_{\si=\pm}\|B^{(\si)}_{n,\mu,\nu}\|_\fm \wf'
         +\GGa_1\|V\|_\fm\wf^2\ \wf'
\end{align*}
First observe that $\wf_\phi$ and $\wf'_\phi$ are each bounded by 
a constant times $\wf$ and $\wf'$, respectively. So for a suitable choice 
of $\rrho_3$, the hypothesis of Lemma \ref{lemBGEdephi} is satisfied. 
The claims concerning $\De\phi_{(*)n}$ and $\De\phi_{(*)n,\nu}$ now follow 
by substituting
\begin{alignat*}{5}
\phi_{(*)}
     &=\phi_{(*)n}(\psi_*,\psi,\mu,\cV)&
     &=S_n(\mu)^{(*)}Q_n^* \fQ_n\,\psi_{(*)}
      +\phi_{(*)n}^{(\ge 3)}(\psi_*,\psi,\mu,\cV)\\
\partial_\nu\phi_*
&=\partial_\nu\phi_{*n}(\psi_*,\!\psi,\mu,\cV)&
&=B^{(+)}_{n,\nu,\mu}\partial_\nu\psi_*
 +\phi_{*n,\nu}^{(\ge 3)}\big(\psi_*,\!\psi, 
                    \partial_\nu\psi_*,\partial_\nu\psi,
                    \mu,\cV\big) \\
\partial_\nu\phi
&=\partial_\nu\phi_n(\psi_*,\psi,\mu,\cV)&
&=B^{(-)}_{n,\nu,\mu}\partial_\nu\psi
  +\phi_{n,\nu}^{(\ge 3)}\big(\psi_*,\psi, 
                    \partial_\nu\psi_*,\partial_\nu\psi,
                    \mu,\cV\big) 
\end{alignat*}
into the conclusions of Lemma \ref{lemBGEdephi}, using  
Proposition \ref{propBGEphivepssoln} and 
\cite[Corollary \corSUBsubstitution]{SUB}.

From \eqref{eqnBGEDephi} we see  
\begin{align*}
   D_n^*\De\phi_*
    &= \de\mu\,\phi_*
       -\big(Q_n^*\fQ_nQ_n-\mu-\de\mu\big) \De\phi_*
       -(\cV'_*+\de\cV'_*)(\Phi_*,\Phi,\Phi_*)\Big|
                     ^{\Phi{(*)}=\phi_{(*)}+\De\phi_{(*)}}
                     _{\Phi{(*)}=\phi_{(*)}} \\ \noalign{\vskip-0.1in}&\hskip2in
       -\de\cV'_*(\phi_*,\phi,\phi_*)\\
 D_n\De\phi
   &=\de\mu\,\phi
       -\big(Q_n^*\fQ_nQ_n-\mu-\de\mu\big)\De\phi
       -(\cV'+\de\cV')(\Phi,\Phi_*,\Phi)\Big|
                     ^{\Phi{(*)}=\phi_{(*)}+\De\phi_{(*)}}
                     _{\Phi{(*)}=\phi_{(*)}}\\ \noalign{\vskip-0.1in}&\hskip2in
       -\de\cV'(\phi,\phi_*,\phi)
\end{align*}
with
$\phi_{(*)}=\phi_{(*)n}(\mu,\cV)$ and 
$\De\phi_{(*)}=\phi_{(*)n}(\mu+\de\mu,\cV+\de\cV)-\phi_{(*)n}(\mu,\cV)$. 
Now just substitute for $\phi_{(*)n}$ using Proposition \ref{propBGEphivepssoln} 
and for $\De\phi_{(*)n}$ using the first part of this proposition. 

\end{proof}

\newpage
\section{The Critical Field}\label{secBGEcritfld}

In this subsection we formulate and prove a precise version of
\cite[Proposition \propHTexistencecriticalfields]{PAR1}. 
Recall from \cite[(\eqnSThatpsi)]{PAR1} that 
\begin{equation}\label{eqnBGEhatpsi}
\hat \psi_{(*)n}(\psi_*,\psi,\mu,\cV)
=\bbbs\big[\psi_{(*)n}(\bbbs^{-1}\psi_*,\bbbs^{-1}\psi,\mu,\cV)\big]
\end{equation}
is a rescaled version of the critical field $\psi_{(*)n}$.

\begin{proposition}\label{propCFpsisoln}
Let $n\ge 1$.
There are constants\footnote{Recall Convention \ref{convBGEconstants}.}
$\GGa_6$, $\rrho_4>0$ 
such that the following hold if 
$
\sfrac{1}{L}\|V\|_{\fm} \wf^2+L^2 |\mu| \le \rrho_4
$
.

There are field maps 
$\hat \psi_{(*)n}^{(\ge 3)}\big(\psi_*,\psi,\mu\big)$ such that
\begin{align*}
\hat \psi_{(*)n}(\psi_*,\psi,\mu,\cV)
     &=\sfrac{a}{L^2}\bbbs C^{(n)}(\mu)^{(*)} Q^*\bbbs^{-1}\psi_{(*)} 
 +\hat \psi_{(*)n}^{(\ge 3)}(\psi_*,\psi,\mu,\cV)
\end{align*}
where
\begin{align*}
C^{(n)}(\mu)&=\big(\sfrac{a}{L^2}Q^*Q+\De^{(n)}(\mu)\big)^{-1} \\
\De^{(n)}(\mu) &= \begin{cases}
          \fQ_n -\fQ_n\,Q_n S_n(\mu) Q_n^* \fQ_n &\text{if $n\ge 1$}\\
          D_0-\mu & \text{if $n=0$}
          \end{cases}
\end{align*}
and 
\begin{equation*}
\tn \hat \psi_{(*)n}\tn\le \GGa_6 \wf\qquad
\TN \hat \psi_{(*)n}^{(\ge 3)} \TN
\le \GGa_6 \sfrac{1}{L}\|V\|_{\fm}\wf^3
\end{equation*}
Furthermore 
$\hat \psi_{(*)n}^{(\ge 3)}$ is of degree at least one in $\psi_{(*)}$ and 
is of degree at least three in $(\psi_*,\psi)$.

\noindent
There are also field maps 
$\hat \psi_{(*)n,\nu}\big(\psi_*,\psi,\psi_{*\nu},\psi_\nu,\mu,\cV\big)$
and $\hat \psi_{(*)n,\nu}^{(\ge 3)}
               \big(\psi_*,\psi,\psi_{*\nu},\psi_\nu,\mu,\cV\big)$
and a linear operator $B_{\psi_{(*)},n,\nu}(\mu)$ such that
\begin{align*}
\partial_\nu\hat \psi_{(*)n}(\psi_*,\psi,\mu,\cV) 
&=\hat \psi_{(*)n,\nu}
        \big(\psi_*,\psi,\partial_\nu\psi_*,\partial_\nu\psi,\mu,\cV\big)\\
&=B_{\psi_{(*)},n,\nu}(\mu)\,\partial_\nu\psi_{(*)}
   +\hat \psi_{(*)n,\nu}^{(\ge 3)}
        \big(\psi_*,\psi,\partial_\nu\psi_*,\partial_\nu\psi,\mu,\cV\big)
\end{align*}
and
\begin{equation*}
\tn \hat \psi_{(*)n,\nu}\tn\le \GGa_6 \wf'\qquad
\TN \hat \psi_{(*)n,\nu}^{(\ge 3)}\TN
          \le \GGa_6\sfrac{1}{L}\|V\|_\fm\wf^2 \wf'
\end{equation*}
Furthermore $\hat \psi_{*n,\nu}^{(\ge 3)}$ and 
$\hat \psi_{n,\nu}^{(\ge 3)}$ are each 
of degree precisely one in $\psi_{(*)\nu}$ and of degree at least two in 
$\big(\psi_*, \psi\big)$.

\end{proposition}
\begin{proof} Set
\begin{equation}\label{eqnBGEcheckSmu}
\check S_{n+1}(\mu)
      =  L^2\, \bbbs^{-1}S_{n+1}(L^2\mu)\bbbs
     =\big\{D_n-\mu
        +\check Q_{n+1}^*\check\fQ_{n+1}\check Q_{n+1}\big\}^{-1}
     :\cH_n\rightarrow\cH_n 
\end{equation}
where, as in \cite[Lemma \lemSCacheckOne]{PAR1}, $\check Q_n = \bbbs^{-1} Q_n\bbbs$
and $\check \fQ_n=\sfrac{1}{L^2}\bbbs^{-1} \fQ_n\bbbs$.
Observe that, by \cite[Remark \remBSedA.e]{BlockSpin}
and the fact that under the substitutions \cite[(\eqnBGAblosckspinSub)]{PAR1},
$\check\fQ=\check\fQ_{n+1}$, 
$\check Q_-=\check Q_{n+1}$ and
$\check S=\check S_{n+1}(\mu)$,
\begin{equation}\label{eqnBGEanotherC}
\sfrac{a}{L^2} C^{(n)}(\mu)^{(*)} Q^*
=\big(\sfrac{a}{L^2}Q^* Q+\fQ_n\big)^{-1}
    \big\{\sfrac{a}{L^2}Q^*  
   +\fQ_n Q_n \check S_{n+1}(\mu)^{(*)}\check Q_{n+1}^* \check\fQ_{n+1}
   \big\}
\end{equation}
By \cite[Definition \defBGAphicheck]{PAR1}
and Proposition \ref{propBGEphivepssoln} with $n$ replaced by $n+1$,
\begin{align*}
\check\phi_{(*)n+1}(\th_*,\th,\mu,\cV)
&= \bbbs^{-1}S_{n+1}(L^2\mu)^{(*)}Q_{n+1}^* \fQ_{n+1}\,\bbbs\th_{(*)}
  +\bbbs^{-1}\!\big[\phi_{(*)n+1}^{(\ge 3)}
           (\bbbs\th_*,\bbbs\th,\!L^2\mu,\bbbs\cV)\big]
\end{align*}
Hence, by the definition of $\psi_{(*)n}$ 
in \cite[Proposition \propBGAomnibus, Lemma \lemSCacheckOne.b]{PAR1}, 
\eqref{eqnBGEcheckSmu} and \eqref{eqnBGEanotherC},
\begin{equation}\label{eqnBGEpsiNexp}
\begin{split}
&\psi_{(*)n}(\th_*,\th,\mu,\cV) 
=\big(\sfrac{a}{L^2}Q^* Q+\fQ_n\big)^{-1}
    \big\{\sfrac{a}{L^2}Q^*\th_{(*)} 
      +\fQ_nQ_n\check\phi_{(*)n+1}(\th_*,\th,\mu,\cV)\big\}\\
&\hskip0.5in=\big(\sfrac{a}{L^2}Q^* Q+\fQ_n\big)^{-1}
    \big\{\sfrac{a}{L^2}Q^*  
   +\fQ_n Q_n \check S_{n+1}(\mu)^{(*)}\check Q_{n+1}^* \check\fQ_{n+1}
   \big\}\th_{(*)}
\\ &\hskip1.75in
   +\big(\sfrac{a}{L^2}Q^* Q+\fQ_n\big)^{-1}
         \fQ_nQ_n\bbbs^{-1}\big[\phi_{(*)n+1}^{(\ge 3)}
           (\bbbs\th_*,\bbbs\th,L^2\mu,\bbbs\cV)\big]\\
&\hskip0.5in=\sfrac{a}{L^2} C^{(n)}(\mu)^{(*)} Q^*\th_{(*)}
   +A_{\psi,\phi}\bbbs^{-1}\big[\phi_{(*)n+1}^{(\ge 3)}
           (\bbbs\th_*,\bbbs\th,L^2\mu,\bbbs\cV)\big]
\end{split}
\end{equation}
where
\begin{equation*}
A_{\psi,\phi}=(\sfrac{a}{L^2}Q^* Q+\fQ_n)^{-1}\fQ_nQ_n
\end{equation*}
So, by \eqref{eqnBGEhatpsi},
\begin{align*}
&\hat \psi_{(*)n}(\psi_*,\psi,\mu,\cV) 
  = \bbbs\big[\psi_{(*)n}(\bbbs^{-1}\psi_*,\bbbs^{-1}\psi,\mu,\cV)\big]\\
&\hskip0.5in=\sfrac{a}{L^2}\bbbs C^{(n)}(\mu)^{(*)} Q^*\bbbs^{-1}\psi_{(*)} 
   + \bbbs A_{\psi,\phi}\bbbs^{-1} 
      \phi_{(*)n+1}^{(\ge 3)}(\psi_*,\psi,L^2\mu,\bbbs\cV)
\end{align*}
Defining
\begin{equation*}
\hat \psi_{(*)n}^{(\ge 3)}(\psi_*,\psi,\mu,\cV) = \bbbs A_{\psi,\phi}\bbbs^{-1} 
      \phi_{(*)n+1}^{(\ge 3)}(\psi_*,\psi,L^2\mu,\bbbs\cV)
\end{equation*}
we have the specified bounds on 
$\hat \psi_{(*)n}(\psi_*,\psi,\mu)$, by 
\cite[Propostion \propPOLmain]{POA}, Proposition \ref{propBGEphivepssoln}.a
and the fact that the kernel, $V^{(s)}$, of $\bbbs\cV$ obeys
\begin{equation*}
\|V^{(s)}\|_\fm\le\sfrac{1}{L}\|V\|_\fm
\end{equation*}
by \cite[Lemma \lemSAscaletoscale.a]{PAR1}.

For $\partial_\nu\psi_{(*)n}$ we use that, 
by \cite[Proposition \propPOLmain.b]{POA},
\begin{align*}
&\partial_\nu\hat \psi_{(*)n}(\psi_*,\psi,\mu,\cV)
=\partial_\nu \sfrac{a}{L^2} \bbbs C^{(n)}(\mu)^{(*)} Q^* \bbbs^{-1}\psi_{(*)} 
   \!+\! \partial_\nu\bbbs A_{\psi,\phi}\bbbs^{-1} 
      \phi_{(*)n+1}^{(\ge 3)}\big(\psi_*,\psi,\!L^2\mu,\bbbs\cV\big)\\
&\hskip0.5in
=\bbbs A_{\psi_{(*)}\th_{(*)}\nu}(\mu)\bbbs^{-1}\ \partial_\nu\psi_{(*)}  
   + \bbbs  A_{\psi,\phi,\nu}\bbbs^{-1}
 \phi_{(*)n+1,\nu}^{(\ge 3)}
    \big(\psi_*,\psi,\partial_\nu\psi_*,\partial_\nu\psi,
            L^2\mu,\bbbs\cV\big)
\end{align*}
Now apply \cite[Proposition \propPOLmain.b]{POA} and, for the second term,
Proposition \ref{propBGEphivepssoln}.b.
\end{proof}

\begin{remark}\label{remCFpsisolnZero}
By  \eqref{eqnBGEhatpsi}, the definition of $\psi_{(*)0}$ 
in \cite[Proposition \propBGAomnibus\ and Definition \defBGAphicheck]{PAR1}, 
we have
\begin{equation*}
\hat\psi_{(*)0}(\psi_*,\psi,\mu,\cV)=\phi_{(*)1}(\psi_*,\psi,L^2\mu,\bbbs\cV)
\end{equation*}
Hence Proposition \ref{propBGEphivepssoln} provides the existence of,
properties of, and bounds on $\hat\psi_{(*)0}$.
\end{remark}

\begin{remark}\label{remCFhatpsitopsi}
\cite[Proposition \propHTexistencecriticalfields]{PAR1} follows from
\cite[Proposition \propBGAomnibus]{PAR1}. To get bounds on $\psi_{(*)n}$,
write, by \eqref{eqnBGEhatpsi},
$
\psi_{(*)n}(\th_*,\th,\mu,\cV)
=\bbbs^{-1}\!\big[\hat \psi_{(*)n}(\bbbs\th_*,\bbbs\th,\mu,\cV)\big]
$
and apply Proposition \ref{propCFpsisoln}.
\end{remark}

\begin{remark}[The complex conjugate of the critical field]\label{remBGEcrfandcomplexconj}
There exists a constant $\GGa_7$
such that the following holds for all $n\ge 1$.
Let $\th(y)$ be a field on $\cX_{-1}^{(n+1)}$ such that\footnote{
Recall that $L_0=L^2$ and $L_\nu=L$ for $\nu=1,2,3$.} 
$|\th(y)| <\sfrac{1}{L^{3/2}} \wf$ and
$|\partial_\nu \th(y)| <\sfrac{1}{L^{3/2}L_\nu} \wf'$ 
for all $y\in \cX_{-1}^{(n+1)}$ and $0\le \nu \le 3$. Then
\begin{equation*}
\big|   \psi_{*n}(\th^*,\th,\mu,\cV)^*(x)- \psi_{n}(\th^*,\th,\mu,\cV)(x)\big|\le
\GGa_7 \wf'   \qquad \qquad {\rm for\ all\ } x \in \cX_0^{(n)}
\end{equation*}
\end{remark}

\begin{proof}
By \cite[Proposition \propBGAomnibus]{PAR1}, 
\begin{align*}
& \psi_{*n}(\th^*,\th,\mu,\cV)^*- \psi_{n}(\th^*,\th,\mu,\cV)\\
& \hskip 4cm=A_{\psi,\phi}\bbbs^{-1}\big[
\phi_{*n+1}^*(\bbbs\th_*,\bbbs\th,L^2\mu,\bbbs\cV)
-\phi_{n+1}(\bbbs\th_*,\bbbs\th,L^2\mu,\bbbs\cV)
\big]
\end{align*}
with $A_{\psi,\phi}$ as after \eqref{eqnBGEpsiNexp}. Now apply 
Remark \ref{remBGEbgfandcomplexconj}.
\end{proof}

\newpage
\appendix
\section{Norms  and a Fixed Point Theorem}\label{appNormsFixedPoint}

We use the terminology ``field map'' to designate an analytic map that assigns to
one or more fields on a finite set $X$ another field on a finite set $Y$.
We assume that $X$ and $Y$ are equipped with volume factors (like the volume
of a fundamental cell in a finite lattice) $\vol_X$ and $\vol_Y$. Then such a field
map $\phi(\psi_1,\cdots,\psi_n)$ has a unique representation as a power
series
\begin{equation*}
\phi(\psi_1,\cdots,\psi_n)(y)
      = \sum_{r_1,\cdots,r_n\ge 0}\hskip-10pt \vol_X^{r_1+\cdots+r_n}
         \sum_{\atop{\vec x_i\in X^{r_i}} {1\le i\le n}}
            \phi_{r_1,\cdots,r_n}\big(y;\vec x_1,\cdots,\vec x_n\big)\,
            \psi_1(\vec x_1) \cdots\psi_n(\vec x_n)
\end{equation*}
where the coefficients $ \phi_{r_1,\cdots,r_n}\big(y;\vec x_1,\cdots,\vec x_n\big)$ are 
invariant under permutations of the components of each vector $\vec x_i$
and where, for $\,\vec x = (x_1,\cdots,x_r) \in X^r\,$ 
we set $\ \psi(\vec x) = \smprod_{i=1}^r \psi(x_i)\,$. 

To measure the size of field maps, we assume that $X$ and $Y$ are both
subsets of a common metric space with metric $d$. As in \cite[\S\secSUBfieldmaps]{SUB},
we introduce norms 
whose finiteness implies  that all the kernels in its power series representation 
are small and decay exponentially  as their arguments separate. 
The norm of $\phi$ with mass $\fm$ and weight factors $\ka_1,\cdots,\ka_n>0$ 
is defined to be 
\begin{equation*}
\tn \phi \tn =  \sum_{r_1,\cdots,r_n\ge 0}
  \big\|  \phi_{r_1,\cdots,r_n}\big\|_\fm\ 
  \smprod_{i=1}^r\ka_i^{r_i}
\end{equation*} 
where
\begin{align*}
\big\| \phi_{r_1,\cdots,r_n}\big\|_\fm
  =\max\big\{ L_\fm( \phi_{r_1,\cdots,r_n})\,,\,R_\fm( \phi_{r_1,\cdots,r_n})\big\}
\end{align*}
and
\begin{align*}
L_\fm( \phi_{r_1,\cdots,r_n})&=\max_{y\in Y}\ 
     \vol_X^{r_1+\cdots+r_n}\hskip-5pt
         \sum_{\atop{\vec x_i\in X^{r_i}} {1\le i\le n}}
     \big|\phi_{r_1,\cdots,r_n}\big(y;\vec x_1,\cdots,\vec x_n\big)\big|
      e^{\fm \tau_d(y,\vec x_1,\cdots,\vec x_n)}\\
R_\fm( \phi_{r_1,\cdots,r_n})&=\max_{x'\in X}
     \max_{ \atop{\atop{1\le j\le n}{r_j\ne 0}}
                {1\le i\le r_j}}\ 
    \vol_Y\hskip-2pt\sum_{y\in Y}\ 
    \vol_X^{r_1+\cdots+r_n-1}\hskip-5pt
           \sum_{\atop{ \atop{\vec x_\ell\in X^{r_\ell}}{1\le\ell\le n} }
           { {(\vec x_j)}_i=x'} }\hskip-3pt
\big|\phi_{r_1,\cdots,r_n}\big(y;\vec x_1,\cdots,\vec x_n\big)\big|
       \\ \noalign{\vskip-0.4in}&\hskip4.15in
       e^{\fm\tau_d(y;\vec x_1,\cdots,\vec x_n)}
\end{align*}
where the tree length $\,\tau_d(x_1,\cdots,x_p)\,$ is the 
minimal length of a tree in the common metric space  that has  
$\,x_1,\cdots,x_p\,$ among its vertices.

The main tool that we  use in the proof of the existence of and bounds
on the background field is \cite[Proposition \propSUBeqnsoln]{SUB}, 
which provides solutions $\vec\ga=\vec\Ga(\vec\al)$
to equations of the form
\begin{equation*}
\vec\ga
=\vec f(\vec\al)+\vec L\big(\vec\al,\vec\ga\big)
   +\vec B\big(\vec\al,\vec\ga\big)
\end{equation*}
Here 
\begin{itemize}[leftmargin=*, topsep=2pt, itemsep=2pt, parsep=0pt]
\item[$\circ$] $\vec f(\vec\al) = \big(f_1(\vec\al),\cdots,f_s(\vec\al) \big)$ 
is an $s$--tuple of field maps with each $f_j(\vec\al)$ mapping
the $r$--tuple of fields $\vec\al= \big(\al_1,\cdots,\al_r \big)$
on $X$ to the field $f_j(\vec\al)$ on $Y$.
\item[$\circ$] $\vec L$ and $\vec B$ are both $s$--tuples of field maps with each 
$j^{\mathrm{th}}$ component  mapping the $(r+s)$--tuple of fields $(\vec\al,\vec\ga)$ 
on $X$ and $Y$ to the field $L_j(\vec\al,\vec\ga)$, respectively $B_j(\vec\al,\vec\ga)$,
on $Y$. 
\item[$\circ$]Each $L_j$ is linear in $\vec\ga$. Each $B_j$ is of degree
at least two and at most $d_{\rm max}$ in $\vec\ga$.
\end{itemize}
For the readers convenience, here is the basic statement of 
\cite[Proposition \propSUBeqnsoln]{SUB}.

\begin{proposition}\label{SUB:propBGEeqnsoln}
Let $\ka_1$, $\cdots$, $\ka_s$ and $\la_1$, $\cdots$, $\la_r$ 
be weight factors for the fields $\al_1,\cdots,\al_s$, on $X$,
and $\ga_1,\cdots,\ga_r$, on $Y$, respectively.
For $s$--tuples of field maps $\vec\Ga(\vec\al) = \big(\Ga_1(\vec\al),\cdots,\Ga_s(\vec\al) \big)$,
we introduce the norm 
\begin{equation*}
\|\vec\Ga\|=\max\limits_{1\le j\le r}\sfrac{1}{\la_j}\tn\Ga_j\tn
\end{equation*}
where $\tn\ \cdot\ \tn$ is the norm with mass $\fm$ and 
weight factors $\ka_1$, $\cdots$, $\ka_s$.
Denote by $\cB_1=\set{\vec\Ga}{\|\vec\Ga\|\le 1}$ the closed unit ball.
\medskip

\noindent
Let $0<\fc<1$. 
Assume that, in the notation above,
\begin{align*}
\tn f_j\tn+\tn L_j\tn+\tn B_j\tn&\le\la_j\\
 \tn L_j\tn+d_{\rm max}\tn B_j\tn&\le \fc\la_j
\end{align*}
 for $1\le j \le r$.
Then there is a unique $\vec\Ga\in\cB_1$ for which
\begin{equation*}
\vec\Ga(\vec\al)
=\vec f(\vec\al)+\vec L\big(\vec\al,\vec\Ga(\vec\al)\big)
   +\vec B\big(\vec\al,\vec\Ga(\vec\al)\big)
\end{equation*}
Furthermore
\begin{equation*}
\max_{j}\sfrac{1}{\la_j}\tn\Ga_j\tn
         \le\sfrac{1}{1-\fc}\max_{j}\sfrac{1}{\la_j}\tn f_j\tn\qquad
\max_{j}\sfrac{1}{\la_j}\tn\Ga_j-f_j\tn
         \le\sfrac{\fc}{1-\fc}\max_{j}\sfrac{1}{\la_j}\tn f_j\tn
\end{equation*}
\end{proposition}

\noindent
There are more refined statements in \cite[Proposition \propSUBeqnsoln]{SUB}.

\newpage
\bibliographystyle{plain}
\bibliography{refs}

\begin{thebibliography}{1}

\bibitem{ParOv}
T.~Balaban, J.~Feldman, H.~Kn{\"o}rrer, and E.~Trubowitz.
\newblock {Complex Bosonic Many--body Models: Overview of the Small Field
  Parabolic Flow}.
\newblock Preprint, 2016.

\bibitem{POA}
T.~Balaban, J.~Feldman, H.~Kn{\"o}rrer, and E.~Trubowitz.
\newblock {Operators for Parabolic Block Spin Transformations}.
\newblock Preprint, 2016.

\bibitem{SUB}
T.~Balaban, J.~Feldman, H.~Kn{\"o}rrer, and E.~Trubowitz.
\newblock {Power Series Representations for Complex Bosonic Effective Actions.
  III. Substitution and Fixed Point Equations}.
\newblock Preprint, 2016.

\bibitem{BlockSpin}
T.~Balaban, J.~Feldman, H.~Kn{\"o}rrer, and E.~Trubowitz.
\newblock {The Algebra of Block Spin Renormalization Group Transformations}.
\newblock Preprint, 2016.

\bibitem{PAR1}
T.~Balaban, J.~Feldman, H.~Kn{\"o}rrer, and E.~Trubowitz.
\newblock {The Small Field Parabolic Flow for Bosonic Many--body Models: Part 1
  --- Main Results and Algebra}.
\newblock Preprint, 2016.

\bibitem{PAR2}
T.~Balaban, J.~Feldman, H.~Kn{\"o}rrer, and E.~Trubowitz.
\newblock {The Small Field Parabolic Flow for Bosonic Many--body Models: Part 2
  --- Fluctuation Integral and Renormalization}.
\newblock Preprint, 2016.

\end{thebibliography}

\end{document}